%% file: query.tex
\documentclass[11pt]{article}
\usepackage{microtype}%if unwanted, comment out or use option "draft"
\usepackage{epsfig}
\usepackage{microtype}
\usepackage{graphicx}
\usepackage{enumitem}
\usepackage[linesnumbered,vlined, ruled]{algorithm2e}
\usepackage{algpseudocode}
\usepackage{epsfig,enumerate,amsmath,amsfonts,amssymb,amsthm,mathrsfs,ifpdf}
\usepackage{indentfirst,relsize}
\usepackage[numbers]{natbib}
\usepackage{setspace}
\usepackage{enumerate}
\usepackage{latexsym}
\usepackage{stackrel}
\usepackage[all]{xy}
\usepackage{fullpage}
\usepackage[usenames,dvipsnames]{pstricks}
\usepackage{pst-grad} % For gradients
\usepackage{pst-plot} % For axes
\usepackage{xspace}
\usepackage{hyperref}

\newcommand{\size}[1]{\left| #1 \right|}

\newcommand{\E}{\mathbb{E}}
\newcommand{\V}{\mathbb{V}}
\newcommand{\remove}[1]{}

\newcommand{\R}{\mathbb{R}}
\newcommand{\N}{\mathbb{N}}

\newcommand{\cT}{\mathcal{T}}

\newcommand{\cE}{\mathcal{E}}

\newcommand{\cB}{\mathcal{B}}
\newcommand{\cD}{\mathcal{D}}
\newcommand{\cP}{\mathcal{P}}

\newcommand{\cG}{\mathcal{G}}

\newcommand{\cO}{\mathcal{O}}
\newcommand{\Oh}{\mathcal{O}}
\newcommand{\tOh}{\widetilde{\mathcal{O}}}
\newcommand{\cF}{\mathcal{F}}

\newcommand{\cH}{\mathcal{H}}
\newcommand{\cN}{\mathcal{N}}
\newcommand{\cX}{\mathcal{X}}
\newcommand{\eps}{\epsilon}

\newcommand{\tis}{{\sc TIS}\xspace}
\newcommand{\bis}{{\sc BIS}\xspace}
\newcommand{\colored}{{\sc COLORED}\xspace}

\newcommand{\test}{{\sc Triangle-Estimation}\xspace}
\newcommand{\pr}{{\mathbb{P}}\xspace}
\newcommand{\verest}{{\sc Verify-Estimate}\xspace}
\newcommand{\cest}{{\sc Coarse-Estimate}\xspace}
\newcommand{\act}{{\sc Active}}

\theoremstyle{plain}
\newtheorem{theo}{Theorem}%[section]
\newtheorem{lem}[theo]{Lemma}
\newtheorem{pro}[theo]{Proposition}

\newtheorem{cl}[theo]{Claim}

\theoremstyle{definition}
\newtheorem{defi}[theo]{Definition}
\newtheorem{rem}{Remark}
\newtheorem{obs}[theo]{Observation}

\newcommand{\defproblem}[3]{
  \vspace{1mm}
\noindent\fbox{
  \begin{minipage}{0.96\textwidth}
  \begin{tabular*}{\textwidth}{@{\extracolsep{\fill}}lr} #1 \\ \end{tabular*}
  {\bf{Input:}} #2  \\
  {\bf{Output:}} #3
  \end{minipage}
  }
  \vspace{1mm}
}

\begin{document}
%\title{On hopelessness of streaming algorithms for geometric problems}
\title{On Triangle Estimation using Tripartite Independent Set
  Queries~\footnote{ A preliminary version of this paper has been
    accepted in ISAAC'19}
\remove{\footnote{The current work on triangle estimation using \tis was the
  first work to estimate a structure more complex than edge in graphs
  using a query oracle modeled on the lines of group queries. Since
  this paper, there have been two other
  works~\cite{BhattaBGM18,DellLM20} that deal with generalizations of
  this work. The significance of the current worklies in the fact that
  it initiated this line of research in estimating structures like
  triangles and beyond in graphs using query oracles modeled along the
  line of group queries, as Beame et al.~\cite{BeameHRRS18} had done
  on edge estimation using \bis.}
}
}

%\titlerunning{Streaming} %optional,
\author{
Anup Bhattacharya \footnote{Indian Statistical Institute, Kolkata, India}
\and 
Arijit Bishnu \footnotemark[2]
\and
Arijit Ghosh \footnotemark[2]
\and
Gopinath Mishra \footnotemark[2]
}
\date{}
\maketitle
%\begin{keyword}
%Data stream, sub-linear, sketch, geometric problems, lower bound, property testing 
%\end{keyword}

\maketitle

\begin{abstract}
\input{abstract}

\end{abstract}
\newpage
%\begin{keyword}
%{Triangle estimation, query complexity, sublinear algorithm}
%\end{keyword}

%\newpage

\input{intro.tex}

\input{overview.tex}
%\input{finalalgo.tex}
\input{triangle.tex}
\input{conclude.tex}

\section*{Acknowledgment}
\noindent
Anup Bhattacharya is supported by NPDF fellowship (No. PDF/2018/002072), Government of India. Arijit Ghosh is supported by Ramanujan Fellowship (No. SB/S2/RJN-064/2015), India. We thank the anonymous reviewers of ISAAC 2019 and TOCS to improve the results and presentation of the paper. 

\newpage
\appendix
\input{appendix.tex}

\bibliographystyle{alpha}
\bibliography{reference}
%\newpage

\end{document}

%% file: abstract.tex
Estimating the number of triangles in a graph is one of the most fundamental problems in sublinear algorithms. In this work, we provide an algorithm that approximately counts the number of triangles in a graph using only polylogarithmic queries when \emph{the number of triangles on any edge in the graph is polylogarithmically bounded}. Our query oracle {\em Tripartite Independent Set} (TIS) takes three disjoint sets of vertices $A$, $B$ and $C$ as inputs, and answers whether there exists a triangle having one endpoint in each of these three sets. Our query model generally belongs to the class of \emph{group queries} (Ron and Tsur, ACM ToCT, 2016; Dell and Lapinskas, STOC 2018) and in particular is inspired by the {\em Bipartite Independent Set} (BIS) query oracle of Beame {\em et al.} (ITCS 2018).\remove{Their algorithm for edge estimation requires only polylogarithmic BIS queries, where a BIS query takes two disjoint sets $X$ and $Y$ as input and answers whether there is an edge with endpoints in $A$ and $B$.} We extend the algorithmic framework of Beame {\em et al.}, with \tis replacing \bis, for approximately counting triangles in graphs. 

\remove{counting using ideas from color coding due to Alon et al. (J. ACM, 1995) and a concentration inequality for sums of random variables with bounded dependency (Janson, Rand. Struct. Alg., 2004).}

\remove{Even though our algorithmic framework is similar to that of Beame {\em et al.}, our analysis for triangle counting is considerably more involved.}

%% file: intro.tex
\section{Introduction} \label{sec:intro} 
\noindent
Counting the number of triangles in a graph is a fundamental algorithmic problem in the RAM model~\cite{AlonYZ97, BjorklundPWZ14, ItaiR78}, streaming~\cite{AhmedDNK14, AhnGM12, Bar-YossefKS02, BuriolFLMS06,  CormodeJ17,   JhaSP13, JowhariG05, KallaugherP17,   KaneMSS12,  PavanTTW13, TangwongsanPT13} and the query model~\cite{ EdenLRS15, GonenRS11}.
In this work, we provide the first approximate triangle counting algorithm using only polylogarithmic queries to a query oracle named \emph{Tripartite Independent Set} (\tis).

\subsection{Notations, the query model, the problem and the result}
\noindent
%\subsection{Notations, the query model, the problem and the result}
%\label{ssec:notation}
%\paragraph*{Notations, the query model, the problem and the result}
We denote the set $\{1,\ldots,n\}$ by $[n]$. Let $V(G)$, $E(G)$ and $T(G)$ denote the set of vertices, edges and triangles in the input graph $G$, respectively. When the graph $G$ is explicit, we may write only $V$, $E$ and $T$ for the set of vertices, edges and triangles.
Let $t(G)=\size{T(G)}$. The statement $A, \, B, \, C$ are disjoint, means $A, \, B, \, C$ are pairwise disjoint. For three non-empty disjoint sets $A, \, B, \, C \subseteq V(G)$, $G(A, \, B, \, C)$, termed as a \emph{tripartite subgraph} of $G$, denotes the induced subgraph of $A \cup B \cup C$ in $G$ minus the edges having both endpoints in $A$ or $B$ or $C$. $t(A,B,C)$ denotes the number of triangles in $G(A,B,C)$. We use the triplet $(a, \, b, \, c)$ to denote the triangle having $a, \, b, \, c$ as its vertices. Let $\Delta_u$ denote the number of triangles having $u$ as one of its vertices. Let $\Delta_{(u,v)}$ be the number of triangles having $(u,v)$ as one of its edges and $\Delta_E = \max_{(u,v) \in E(G)}\Delta_{(u,v)}$. For a set $\cal U$, ``$\cal U$ is \colored with [n]'', means that each member of $\cal U$ is assigned a color out of $[n]$ colors independently and uniformly at random. 
Let $\E[X]$ and $\V[X]$ denote the expectation and variance of a random variable $X$. For an event $\cE$, $\cE^c$ denotes the complement of $\cE$. The statement ``$a$ is an $(1 \pm \eps)$-multiplicative approximation of $b$'' means $\size{b-a} \leq \eps \cdot b$. 
Next, we describe the query oracle.
\begin{defi}[Tripartite independent set oracle (\tis)]
%\item[Tripartite independent set oracle (\tis)] 
Given three non-empty disjoint subsets $V_1,V_2,V_3 \subseteq V(G)$ of a graph $G$, \tis{} query oracle answers `YES' if and only if $t(V_1, \, V_2, \, V_3) \neq 0$.
\end{defi}
Notice that the query oracle looks at only those triangles that have vertices in all of these sets $V_1, \, V_2, \, V_3$.
%\paragraph*{Problem definition and result}
%\vspace{-0.1in}
\noindent 
\remove{The problem definition and our main result are given as follows.}
The \test problem is to report an $(1 \pm \eps)$-multiplicative approximation of $t(G)$ where the input is $V(G)$, \tis oracle for graph $G$ and $\eps\in (0,1)$.
\remove{
\defproblem{\test}{Set of vertices $V(G)$, \tis oracle for graph $G$ and $\eps\in (0,1)$.}{$1 \pm \eps$ multiplicative approximation of $t(G)$.}
}
\begin{theo}[Main result]
\label{theo:main-restate}
{
Let $G$ be a graph with $\Delta_E \leq d$, $\size{V(G)}=n \geq 64$. For any
$\eps >0$, \test can be solved using $\cO\left( \frac{d^{2}\log ^ {18} n}{\eps^{4}} \right)$ many \tis queries with probability {at least} $1-\frac{\cO(1)}{n^{2}}$.
}
\end{theo}

Note that the query complexity stated in Theorem~\ref{theo:main-restate} is $\mbox{poly}(\log n, \frac{1}{\eps})$, even if $d$ is $\Oh(\log^c n)$, where $c$ is a positive constant. We reiterate that the only bound we require is on the number of triangles on an edge; neither do we require any bound on the maximum degree of the graph, nor do we require any bound on the number of triangles incident on a vertex.

\subsection{Query models and \tis}
%\subsection{Query models and \tis}
%\label{ssec:querymodel-tis}
%\paragraph*{Query models and \tis}
\noindent
Query models for graphs are essentially of two types: Local Queries, and  Group (and related) queries. 

%\begin{description}
\paragraph{\em Local Queries.} This query model was initiated by Feige~\cite{Feige06} and Goldreich and Ron~\cite{GoldreichR08} and used even recently by~\cite{EdenLRS15,EdenRS18}. The queries on the graphs are 
\begin{description}
  \item[(i)] {Degree query:} the oracle reports the degree of a vertex; 

  \item[(ii)] {Neighbor query:} the oracle reports the $i^{\mbox{\tiny{th}}}$ neighbor of $v$, if it exists; and 

  \item[(iii)] {Edge existence query:} the oracle reports whether there exists an edge between a given pair of vertices.
\end{description}

\paragraph{\em Group Queries or Subset queries and Subset samples.} These queries were implicitly initiated in the works of Stockmeyer~\cite{Stockmeyer83,Stockmeyer85} and formalized by Ron and Tsur~\cite{RonT16}. \emph{Group queries} can be viewed as a generalization of membership queries in sets. The essential idea of the group queries is to estimate the size of an unknown set $S \subseteq U$ by using a YES/NO answer from the oracle to the existence of an intersection between sets $S$ and $T \subseteq U$; and give a uniformly selected item of $S \cap T$, if $S \cap T \not= \emptyset$ in the \emph{subset sample} query. Subset sample queries are at least as powerful as group queries. The \emph{cut query} by Rubinstein et al.~\cite{RubinsteinSW18}, though motivated by submodular function minimization problem, can also be seen in the light of group queries --- we seek the number of edges that intersect both the vertex sets that form a cut. Choi and Kim~\cite{ChoiK08} used a variation of group queries for \emph{graph reconstruction}. Dell and Lapinskas~\cite{DellL18} essentially used this class of queries for estimating the number of edges 
in a bipartite graph. 
{\em Bipartite independent set} (\bis) queries for a graph, initiated by Beame et al.~\cite{BeameHRRS18}, can also be seen in the light of group queries. It provides a YES/NO answer to the existence of an edge in $E(G)$ that intersects with both $V_1, V_2 \subset V(G)$ of $G$, where $V_1$ and $V_2$ are disjoint. A subset sample version of \bis oracle was used in~\cite{BishnuGKM018}. 
%\end{description}

In \tis, we seek a YES/NO answer about the existence of an intersection between the set of triangles, that we want to estimate, and three disjoint sets of vertices. Thus \tis belongs to the class of group queries, as does \bis. 
A bone of contention for any newly introduced query oracle is its worth\footnote{See \url{http://www.wisdom.weizmann.ac.il/~oded/MC/237.html} for a comment on \bis.}. Beame {\em et al.}~\cite{BeameHRRS18} had given a subjective justification in favor of \bis to establish it as a query oracle. It is easy to verify that \tis, being in the same class of group queries, have the interesting connections to group testing and computational geometry as \bis. We provide justifications in favor of considering $\Delta_E  \leq d$ in Appendix~\ref{sec:deltabound}. Intuitively, \tis is to triangle counting what \bis is to edge estimation.

\subsection{Prior works}
\noindent
Eden et al.~\cite{EdenLRS15} showed that query complexity of estimating the number of triangles in a graph $G$ using local queries is
$\widetilde{\Theta}\left( \frac{\size{V(G)}}{t(G)^{1/3}} + \min \left\lbrace \frac{\size{E(G)}^{3/2}}{t(G)}, \size{E(G)} \right\rbrace\right)$\footnote{$\tOh(\cdot)$ hides a polynomial factor of $\log n$ and $\frac{1}{\eps}$, where $\eps \in (0,1)$ is such that $(1-\eps) t \leq \hat{t} \leq (1+\eps)t$; $\hat{t}$ and $t$ denote the estimated and actual number of triangles in $G$, respectively.}.
Matching upper and lower bounds on $k$-clique counting in $G$ using local query model have also been reported~\cite{EdenRS18}. 
These results have almost closed the $k$-clique counting problem in graphs using local queries. 
A precursor to triangle estimation in graphs is edge estimation. 
The number of edges in a graph $G$ can be estimated by using $\tOh\left( \frac{\size{V(G)}}{\sqrt{\size{E(G)}}}\right)$ many degree and neighbor 
queries, and $\Omega\left(\frac{\size{V(G)}}{\sqrt{\size{E(G)}}}\right)$ queries are necessary to estimate the number of edges even if we allow all the three local queries~\cite{GoldreichR08}. 
This result would almost have closed the edge estimation problem but for having a relook at the problem with stronger query models and hoping for polylogarithmic number of queries. 
Beame {\em et al.}~\cite{BeameHRRS18} precisely did that by estimating the number of edges in a graph 
using $\Oh \left( \frac{\log ^{14} n}{\eps^{4}} \right)$ {\em bipartite independent set} (\bis) queries. 
Motivated by this result, we explore whether triangle estimation can be solved using only polylogarithmic \tis queries.

Note that the \test can also be thought of as {\sc Hyperedge Estimation} problem in a $3$-uniform hypergrah. As a follow up to this paper, Dell et al.~\cite{DellLM20} and Bhattacharya et al.~\cite{BhattaBGM18}, independently, generalized our result to $c$-uniform hypergraphs, where $c \in \N$ is a constant. Their result showed that the bound on $\Delta_E$ is not necessary to solve \test by using polylogarithmic many \tis queries.

\subsection*{Organization of the paper}
\noindent
We give a broad overview of the algorithm in Section~\ref{sec:algooverview}. Section~\ref{sec:sparse} gives the details of  sparsification. In Section~\ref{sec:coarse}, we give exact/approximate estimation algorithm with respect to a threshold. Section~\ref{sec:coarse_new}  discusses about teh algorithm for \emph{coarse} estimation of the number of triangles. The final algorithm is given in Section~\ref{sec:finaltrialgo-app}. {Section~\ref{sec:conclude} concludes the paper with some discussions about future improvements.} Appendix~\ref{sec:deltabound} provides justifications in favor of \tis. Appendix~\ref{sec:prelim} has the probabilistic results used in this paper.

%% file: overview.tex
\section{Overview of the algorithm}
\label{sec:algooverview}
\noindent 
Our algorithmic framework \remove{uses \emph{sparsification}, \emph{subgraph sampling} and \emph{coarse estimation} as in} is inspired by~\cite{BeameHRRS18} but the detailed analysis is markedly different, like the use of 
%using {\em color coding} due to Alon et al.~\cite{AlonYZ95} and
a relatively new concentration inequality, due to Janson~\cite{Janson04}, for handling sums of random variables with bounded dependency. 
Apart from Lemmas~\ref{lem:exact} and~\ref{lem:importance}, all other proofs require different ideas. 

\remove{
The analysis in~\cite{BeameHRRS18} does not go through in our case because of a subtle difference between counting the number of triangles and the number of edges in a graph. An edge is an  explicit structure, whereas, a triangle is an implicit structure --- triangles $(a,b,x)$, $(b,c,y)$ and $(a,c,z)$ in $G$ imply the existence of the triangle $(a,b,c)$ in $G$. (Gopi: Do we need to mention this?)
}

\begin{figure}[h!]
  \centering
  \includegraphics[width=1\linewidth]{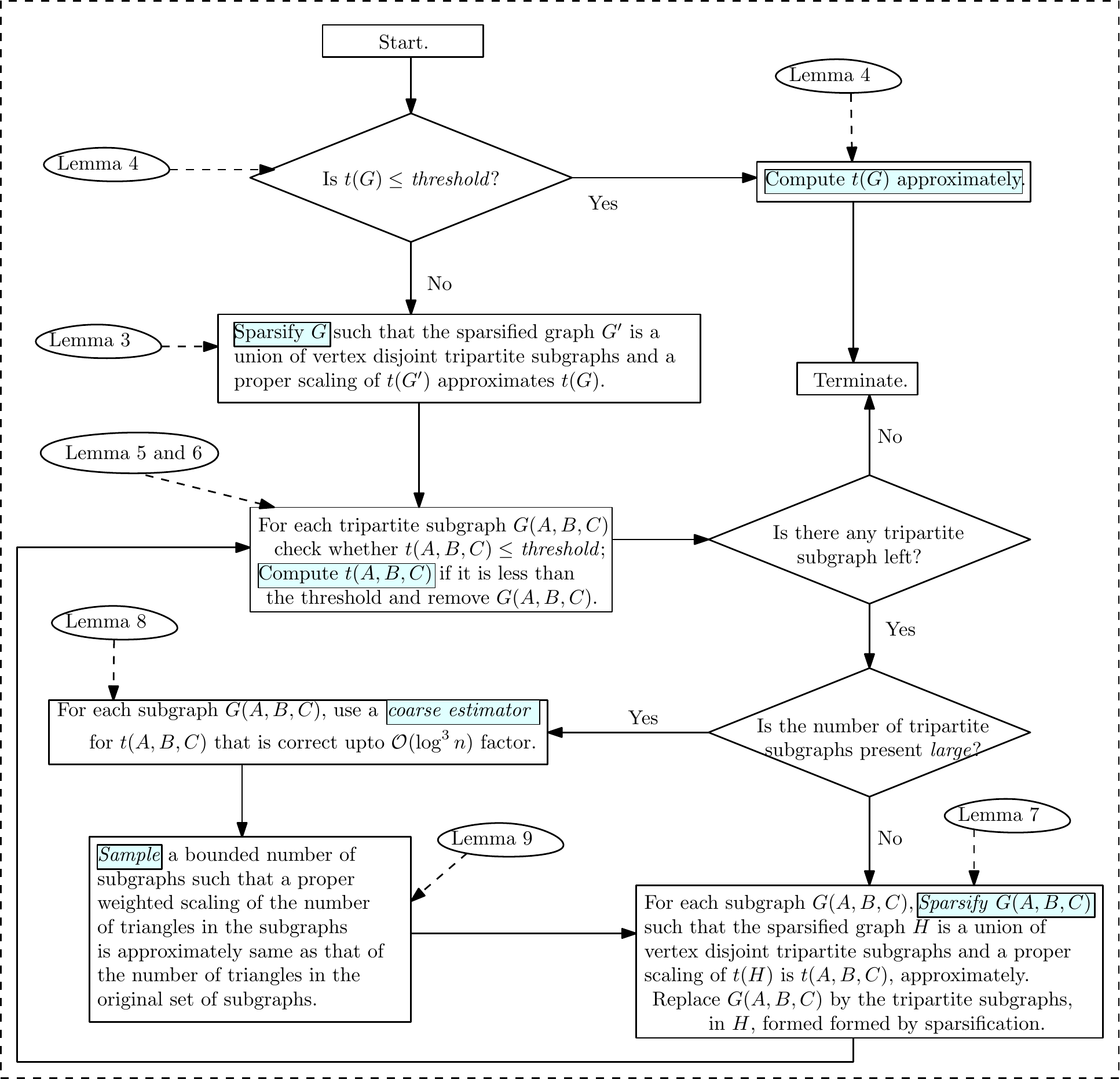}
  {\caption{Flow chart of the algorithm. The highlighted texts indicate the basic building blocks of the algorithm. We also indicate the corresponding lemmas that support the building blocks.}}
  \label{fig:flowchart}
\end{figure}
In Figure~\ref{fig:flowchart}, we give a \emph{flowchart} of the algorithm and show the corresponding 
lemmas that support the steps of the algorithm. 
The main idea of our algorithm is as follows. We can figure out for a given $G$, if the number of triangles $t(G)$ 
is greater than a threshold $\tau$ ((Lemma~\ref{lem:exact_decide})). If $t(G) \leq \tau$, i.e., $G$ is sparse in triangles, we compute an $(1\pm \eps)$-approximation of $t(G)$ (Lemma~\ref{lem:exact_decide}). Otherwise, we sparsify $G$ to get a disjoint union of tripartite subgraphs 
of $G$ that maintain $t(G)$ up to a scaling factor (Lemma~\ref{theo:sparse}). For each tripartite subgraph, if the 
subgraph is sparse (decided by Lemma~\ref{lem:decide_exact}), we count the number of triangles exactly 
(Lemma~\ref{lem:exact}). Otherwise, we again sparsify (Lemma~\ref{theo:sparse1}). This repeated process of 
sparsification may create a huge number of tripartite subgraphs. Counting the number of triangles in them is 
managed by doing a coarse estimation (Lemma~\ref{lem:coarse_main}) and taking a sample of the subgraph 
that maintains the number of triangles \emph{approximately}. Each time we sparsify, we ensure that the sum 
of the number of triangles in the subgraphs generated by sparsification is a constant fraction of the number 
of triangles in the graph before sparsification, making the number of iterations $O(\log n)$.

\remove{
The basic building blocks of the algorithm are: two kinds of \emph{sparsification} routines (one for general graph and another for tripartite graph), a \emph{coarse estimator}, a \emph{sampling} scheme of the subgraphs and two algorithms for exactly counting the number of triangles (one for general graph and another for tripartite graph) when the number of triangles is not too \emph{large}.}

 We sparsify $G$ by considering the partition obtained when $V(G)$ is COLORED with $[3k]$. This sparsification is done such that:  (i) the sparsified graph is a union of a set of vertex disjoint tripartite subgraphs and  (ii) a proper scaling of the number of triangles in the sparsified graph is a \emph{good} estimate of $t(G)$ with high probability\footnote{High probability means that the probability of success is at least $1-\frac{1}{n^c}$ for some constant $c$.}. The proof of the sparsification result stated next uses the \emph{method of averaged bounded differences} and \emph{Chernoff-Hoeffding type} inequality in bounded dependency setting by Janson~\cite{Janson04}. The detailed proof is in Section~\ref{sec:sparse}. Recall that $\Delta_E$ is the maximum number of triangles on a particular edge. 
\begin{lem}[General Sparsification]
\label{theo:sparse} Let $k,d \in \N$. There exists a constant $\kappa_1$ such that for any graph $G$ with $\Delta_E \leq d$, if $V_1,\ldots,V_{3k}$ is a random partition of $V(G)$ obtained by 
$V(G)$ being \colored with $[3k]$, then 
%\begin{center}  
$$  
	\pr\left(\size{\frac{9k^2}{2}\sum\limits_{i=1}^kt(V_i,V_{k+i},V_{2k+i}) - t(G)} 
		>  \kappa_1 d k^2 \sqrt{t(G) } \log n \right)~\leq~\frac{2}{n^4}.
$$
%\end{center}
\end{lem}
\remove{
The above tells us that a proper scaling of the number of triangles, in the sparsified graph, approximately estimates $t(G)$, when $t(G)$ is above a threshold~\footnote{The threshold is a fixed polynomial in $d, \log n$ and $\frac{1}{\eps}$.}.} We apply the sparsification corresponding to Lemma~\ref{theo:sparse} only when $t(G)$ is above a threshold\footnote{The threshold is a fixed polynomial in $d, \log n$ and $\frac{1}{\eps}$.} to ensure that the relative error is bounded. We can decide whether $t(G)$ is at most the threshold and if it is so, we \emph{estimate} the value of $t(G)$, using the following lemma, whose proof is given in Section~\ref{sec:coarse}.
{\begin{lem}[Estimation with respect to a threshold] \label{lem:exact_decide}
 {There exists an algorithm that for any graph $G$, a threshold parameter $\tau \in \N$ and an $\eps \in (0,1)$, determines whether $t(G) > \ \tau$. If $t(G) \leq \ \tau$, the algorithm gives a $(1 \pm \eps)$-approximation to $t(G)$ by using $\cO(\frac{\tau \log^2 n}{\eps^2})$ many \tis queries with probability at least $1- n^{-10}$.} %Moreover, the algorithm finds the exact value of $t(G)$ if $t(G) < \tau$.}
\end{lem}}
Assume that $t(G)$ is \emph{large}~\footnote{Large refers to a fixed polynomial in $d, \log n$ and $\frac{1}{\eps}$} and $G$ has undergone sparsification. We initialize a data structure with a set of vertex disjoint tripartite graphs that are obtained after the sparsification step. For each tripartite graph $G(A,B,C)$ in the data structure, we check whether $t(A,B,C)$ is less than a threshold using the algorithm corresponding to Lemma~\ref{lem:decide_exact}. If it is less than a threshold, we compute the exact value of $t(A,B,C)$ using Lemma~\ref{lem:exact} and remove $G(A,B,C)$ from the data structure. The proofs of Lemmas~\ref{lem:decide_exact} and~\ref{lem:exact} are given in Section~\ref{sec:coarse}.
\begin{lem}[Threshold for Tripartite Graph] \label{lem:decide_exact}
There exists a deterministic algorithm that given any disjoint subsets $A,B,C \subset V(G)$ of any graph $G$ and a threshold parameter $\tau \in \N$, can decide whether $t(A,B,C) \leq \tau$ using $\Oh(\tau \log n)$ \tis queries. 
\end{lem}

\begin{lem}[Exact Counting in Tripartite Graphs]
\label{lem:exact}
There exists a deterministic algorithm that given any disjoint subsets $A,B,C \subset V(G)$ of any graph $G$, can determine the exact value of $t(A,B,C)$ using $\Oh(t(A,B,C) \log n)$ \tis queries. 
\end{lem}
Now we are left with some tripartite graphs such that the number of triangles in each graph is more than a threshold. If the number of such graphs is not large, then we sparsify each tripartite graph $G(A,B,C)$ in a fashion almost similar to the earlier sparsification.
\remove{
such that (i) the sparsified graph is a
disjoint union of some tripartite subgraphs and (ii) a proper constant scaling of the number of triangles in the sparsified graph is approximately same as that of $t(A,B,C)$.} 
This sparsification result formally stated 
in the following Lemma, has a proof similar to Lemma~\ref{theo:sparse}.  We replace $G(A,B,C)$ by a constant (say, $k$)~\footnote{In our algorithm, $k$ is a constant. However, Lemma~\ref{theo:sparse1} holds for any $k \in \N$.} many tripartite subgraphs formed after sparsification.
\begin{lem}[Sparsification for Tripartite Graphs]
\label{theo:sparse1}
Let $k,d \in \N$. There exists a constant $\kappa_2$ such that 
$$
	\pr\left(\size{{k^2} \sum\limits_{i=1}^k t(A_i,B_i,C_i) - t(A,B,C)} >  \kappa_2 d  k^2 \sqrt{t(G) } \log n \right) \leq \frac{1}{n^8}
$$ 
where $A$, $B$ and $C$ are disjoint subsets of $V(G)$ for any graph $G$ with $\Delta_E \leq d$, and $A_1,\ldots,A_{k}$, $B_1,\ldots,B_{k}$ and $C_1,\ldots,C_{k}$ are the partitions of $A,B,C$ formed uniformly at random, respectively.
\end{lem}
If we have a large number of vertex disjoint tripartite subgraphs of $G$ and each subgraph contains a large  number of triangles, then we \emph{coarsely} estimate the number of triangles in each subgraph which is correct up to $\Oh(\log^3 n)$ factor by using the algorithm corresponding to the following Lemma, whose proof is in Section~\ref{sec:coarse_new}. Our \cest algorithm is similar in structure to the coarse estimation algorithm for edge estimation, but requires 
a more careful analysis.
%the analysis involves sophisticated calculations. 
 \begin{lem}[Coarse Estimation]
\label{lem:coarse_main}
There exists an algorithm that given disjoint subsets $A,B,C \subset V(G)$ of any graph $G$, returns an estimate $\widetilde{t}$ satisfying 
$$
	\frac{t(A,B,C)}{64 \log^2 n} \leq \widetilde{t} \leq 64 t(A,B,C) \log^2 n
$$
with probability at least $1-n^{-9}$. Moreover, the query 
complexity of the algorithm is $\Oh(\log ^4 n)$.
\end{lem} 
After estimating the number of triangles in each subgraph coarsely, we approximately maintain the triangle count using the following sampling result which is a direct consequence of the {\em Importance Sampling Lemma} of~\cite{BeameHRRS18}.
\footnote{For the exact statement of the Importance Sampling Lemma see Lemma~\ref{lem:importance1} in Appendix~\ref{sec:prelim}.}
%generate a bounded number of samples of the set of subgraphs using a sampling technique given by Beame et al.~\cite{BeameHRRS18}. The sampling maintains the triangle count approximately.  After getting the sample, we apply the sparsification algorithm corresponding to Lemma~\ref{theo:sparse1} for each subgraph in the sample.
%The Lemma corresponding to sampling is stated in Lemma~\ref{lem:importance}. 
%The original statement of 
%Beame et al.~\cite{BeameHRRS18} is given in Lemma~\ref{lem:importance1} in Appendix~\ref{sec:prelim}.
\begin{lem}[\cite{BeameHRRS18}]
\label{lem:importance}
Let $(A_1,B_1,C_1,w_1),\ldots,(A_r,B_r,C_r,w_r)$ be the tuples present in the data structure and $e_i$ be the corresponding coarse estimation for $t(A_i,B_i,C_i), i \in [r],$  such that 
\begin{description}
\item[(i)] $\forall i \in [r]$, we have $w_i$, $e_i \geq 1$;

\item[(ii)] $\forall i \in [r]$, we have $\frac{e_i}{\rho} \leq t(A_i,B_i,C_i) \leq e_i \rho$ for some $\rho >0$; and

\item[(iii)] $\sum_{i=1}^r {w_i\cdot t(A_i,B_i,C_i)} \leq M$.
\end{description}
Note that the exact values $t(A_i,B_i,C_i)$'s are not known to us. Then there exists an algorithm that finds 
 $(A'_1,B'_1,C'_1,w'_1),
\ldots,(A'_s,B'_s,C'_s,w'_s)$ such that all of the above three conditions hold and  
$$
\size{\sum\limits_{i=1}^s {w'_i\cdot t(A'_i,B'_i,C'_i) } - \sum\limits_{i=1}^r {w_i\cdot t(A_i,B_i,C_i)}} \leq \lambda S,
$$ 
with probability at least $1-\delta$, where  $S=\sum_{i=1}^r {w_i\cdot t(A_i,B_i,C_i)}$ and $\lambda, \delta >0$. 
Also, 
$$
s=\cO\left(\lambda^{-2} \rho^4 \log M \left(\log \log M + \log \frac{1}{\delta}\right)\right).
$$
\end{lem} 

Now again, for each tripartite graph $G(A,B,C)$, we check whether $t(A,B,C)$ is less than a threshold using the algorithm corresponding to Lemma~\ref{lem:decide_exact}. If yes, then we can compute the exact value of $t(A,B,C)$ using Lemma~\ref{lem:exact} and remove $G(A,B,C)$ from the data structure. Otherwise, we iterate on all the required steps discussed above as shown in Figure~\ref{fig:flowchart}. Observe that each iteration uses polylogarithmic~\footnote{Polylogarithmic refers to a polynomial in $d, \log n$ and $\frac{1}{\eps}$} many queries. Now, note that the number of triangles reduces 
by a constant factor after each sparsification step. So, the number of iterations is bounded by $\Oh(\log n)$. Hence, the query complexity of our algorithm is polylogarithmic.
This completes the high level description of our algorithm. 

\remove{
We will discuss the technical lemmas next. We will come back to the algorithm again in Section~\ref{sec:finaltrialgo} in a more detailed fashion.}

%% file: triangle.tex
\input{sparsification.tex}

\input{estimate.tex}
\input{coarse.tex}
\input{finalalgo.tex}

%% file: sparsification.tex
%\section{Algorithm for counting the number of triangles}
%\label{sec:triangle}

\section{Sparsification step}
\label{sec:sparse}
\noindent 
In this Section, we prove Lemma~\ref{theo:sparse}. The proof of Lemma~\ref{theo:sparse1} is similar.

\begin{lem}[Lemma~\ref{theo:sparse} restated] 
Let $k,d \in \N$. There exists a constant $\kappa_1$ such that for any graph $G$ with $\Delta_E \leq d$, if $V_1,\ldots,V_{3k}$ is a random partition of $V(G)$ obtained by 
$V(G)$ being \colored with $[3k]$, then 
%\begin{center}  
$$  
	\pr\left(\size{\frac{9k^2}{2}\sum\limits_{i=1}^kt(V_i,V_{k+i},V_{2k+i}) - t(G)} 
		>  \kappa_1 d k^2 \sqrt{t(G) } \log n \right)~\leq~\frac{2}{n^4}.
$$
%\end{center}
\end{lem}
\begin{proof}
\remove{Let us assign a color to  each vertex of the graph independently and uniformly at random out of $[3k]$ colors $\{1,\ldots,3k\}$.} $V(G)$ is \colored with $[3k]$. Let $V_1,\ldots,V_{3k}$ be the resulting partition of $V(G)$. Let $Z_i$ be the random variable that denotes the color assigned to the $i^{\small{th}}$ vertex. For $i \in [3k]$, $\pi(i)$ is a set of three colors defined as follows: $\pi(i)=\{i, (1 + (i+k-1)~\mbox{mod}~3k), (1 + (i+2k-1)~\mbox{mod}~3k)\}$.
\begin{defi}
\label{defi:proper}
A triangle $(a,b,c)$ is said to be \emph{properly colored} if there exists a bijection in terms of coloring from $\{a,b,c\}$ to $\pi(i)$.
\end{defi}
Let $f(Z_1,\ldots,Z_n)= \sum_{i=1}^k t(V_i,V_{k+i},V_{2k+i}).$ Note that $f$ is the number of triangles that are properly colored.
The probability that a triangle is properly colored is $\frac{2}{9k^2}$. So,  $\E[f]=\frac{2t(G)}{9k^2}$.

\remove{Let $f_{ij}(X_1,\ldots,X_n)=t(\{i\},\{j\},V_\Phi(i,j))$ and $f_{i}(X_1,\ldots,X_n)=\sum\limits_{(i,j) \in E(G)} f_{ij}
(X_1,\ldots,X_n)$. Note that $f_{ij}=f_{ji}$ and $f=\frac{1}{3}\sum\limits_{1\leq i < j \leq n} f_{ij} =\frac{1}{3}
\sum\limits_{i \in [n]} f_i$.} 

Let us focus on the instance when vertices $1,\ldots,t-1$ are already colored and we are going to color vertex $t$.
Let $S_{\ell}$~$(S_{r})$ be the set of triangles in $G$ having $t$ as one of the vertices and other two vertices are from $[t-1]$~$([n] \setminus [t])$. $S_{\ell r}$ be the set of triangles in $G$ such that $t$ is a vertex and the second and third 
vertices are from $[t-1]$ and $[n] \setminus [t]$, respectively. \remove{$t$ as one of the vertices, one vertex is from $\{1,\ldots,t-1\}$ and another from $\{t+1,\ldots,n\}$.}
\remove{
\begin{eqnarray*}
&& S_{\ell}=\left\lbrace (i,j,t):1\leq i,j \leq t-1~\mbox{ and $(i,j,t)$ is a triangle in G}\right\rbrace; \\
&& S_{r}=\{(t,i,j):t+1\leq i,j \leq n ~\mbox{ and $(t,i,j)$ is a triangle in G}\};\\
 && S_{\ell r}=\{(i,t,j):1\leq i \leq t-1, t+1 \leq j \leq n~\mbox{ and $(i,t,j)$ is a triangle in G}\}.
 \end{eqnarray*}}
 
 Given that the vertex $t$ is colored with color $c \in [3k]$, let $N_{\ell}^c, N_{r}^c, N_{\ell r}^c$ be the random variables that denote the number of triangles in $S_{\ell}, S_{r}$ and  $S_{\ell r}$ that are properly colored, respectively. Also, let $\E_f^t$ denote the absolute difference 
in the conditional expectation of the number of triangles that are properly colored whose $t^{\small{th}}$-vertex is (possibly) differently colored.
By considering the vertices in $S_{\ell},S_r$ and $S_{\ell r}$ separately, we can bound $\E_f^t$
\begin{eqnarray*}
 \E_f^t &=& \size{\E\left[ f~|~Z_1,\ldots,Z_{t-1}, Z_t=a_t \right]-\E\left[ f~|~Z_1,\ldots,Z_{t-1}, Z_t=a'_t \right]} \\
&=& \size{N_{\ell}^{a_t}-N_{\ell}^{a_t'}  + \E \left[N_{r}^{a_t}-N_{r}^{a_t'}\right] + \E \left[N_{\ell r}^{a_t}-N_{\ell r}^{a_t'}\right]  } \\
&\leq& \size{N_{\ell}^{a_t}-N_{\ell}^{a_t'}}  + \E \left[\size{N_{r}^{a_t}-N_{r}^{a_t'}}\right] + \E \left[\size{N_{\ell r}^{a_t}-N_{\ell r}^{a_t'}}\right]    
\end{eqnarray*}
Now, consider the following claim, which we prove later.
\begin{cl}
 \label{clm:inter}
 \begin{itemize}
 \item[(a)] $\pr(\mid{N_{\ell}^{a_t}-N_{\ell}^{a_t'}}\mid < 8 \sqrt{d \Delta_t \log n}) \geq 1-4n^{-8}$; 
 \item[(b)] $\E [\mid{N_{r}^{a_t}-N_{r}^{a_t'}}\mid ] \leq {\sqrt{d\Delta_t}}/{k}$;
 \item[(c)]  $\E [\mid {N_{\ell r}^{a_t}-N_{\ell r}^{a_t'}}\mid ] < 6d\sqrt{\Delta_t \log n}$.~\footnote{Note that $\Delta_t$ is the number of triangles having $t$ as one of its vertices and we are not assuming any bound on $\Delta_t$. We assume $\Delta_E$, that is number of triangles on any edge, is bounded.}
 \end{itemize}

\end{cl}
Let $c_t=15d \sqrt{{\Delta_t \log n}}$. From the above claim, we have 

$$\E_f^t < 8 \sqrt{d \Delta_t \log n} + \frac{\sqrt{d\Delta_t}}{k} + 6d\sqrt{{\Delta_t \log n}} \leq 15d \sqrt{{\Delta_t \log n}}=c_t $$ 
with probability at least $1-\frac{4}{n^8}$. 
Let $\cB$ be the event that there exists $t \in [n]$ such that $\E_f^t > c_t$.
% $$\size{\E\left[ f~|~X_1,\ldots,X_{t-1}, X_t=a_t \right]-\E\left[ f~|~X_1,\ldots,X_{t-1}, X_t=a'_t \right]} > c_t.$$
 By the union bound over all $t \in [n]$,  $\pr(\cB) \leq  \frac{4}{n^{7}}$.
 
  Using the method of \emph{averaged bounded difference}~\cite{DubhashiP09} (See Lemma~\ref{lem:dp} in Appendix~\ref{sec:prelim}), we have 
 $$
 \pr\left(\size{f - \E[f]} > \delta + t(G)\pr(\cB) \right) \leq e^{-{\delta^2}/{\sum\limits_{t=1}^{n} c_t^2}} + \pr(\cB).
 $$
We set $\delta = 60d\sqrt{t(G)}\log n$. Observe that
  $\sum\limits_{t=1}^{n} c_t^2={225d^2\log n} \sum\limits_{t=1}^n \Delta_t ={675d^2t(G) \log n}$. Hence, 
  $$
  	\pr\left(\size{f - \frac{2t(G)}{9k^2}} >  60d\sqrt{{t(G) }}\log n + t(G)\pr(\cB) \right) \leq \frac{1}{n^4} + \frac{1}{n^7},
  $$ 
  that is,  
$$\pr\left(\size{\frac{9k^2}{2}f - t(G)} >  270dk^{2}\sqrt{t(G) }\log n + \frac{9k^2}{2} \cdot \frac{t(G)}{n^7}\right) \leq \frac{1}{n^4} + \frac{1}{n^7}.$$
Since, $\frac{9k^2}{2} \cdot \frac{t(G)}{n^7} < dk^{2}\sqrt{t(G) }\log n$, we get
$$\pr\left(\size{\frac{9k^2}{2}f - t(G)} >  271dk^{2}\sqrt{t(G) } \log n \right) \leq \frac{2}{n^4}.$$
\end{proof}
To finish the proof of Lemma~\ref{theo:sparse}, we need to prove Claim~\ref{clm:inter}.\remove{ When we fix the color of a veretx of a triangle  Informally speaking, Claim~\ref{clm:inter} bounds the difference between the number of properly colored triangles incident on $t$ when $Z_t=a_t$ and that when $Z_t=a_t'$. } For that, we need the following definition and intermediate result (Lemma~\ref{lem:col1}) that is stated in terms of objects, which in the current context can be thought of as vertices. 
\begin{defi}
\label{defi:col_inter}
Let  $\cX$ be a set of $u$ objects \colored with $[3k]$. Let $\alpha , \beta \in [3k]$ and $\alpha \neq \beta$. A pair of objects $\{a,b\}$ is said to be colored with $\{\alpha,\beta\}$ if there is a bijection in terms of coloring from $\{a,b\}$ to $\{\alpha,\beta\}$. An object $o \in \cX$ is colored with 
$\{\alpha,\beta\}$ if $o$ is colored with $\alpha$ or $\beta$.
\end{defi}
Recall Definition~\ref{defi:proper}. A triangle incident on $t$ is properly colored if the pair of vertices in the triangle other than $t$, is colored with $\pi(Z_t) \setminus \{Z_t\}$. Note that,  Claim~\ref{clm:inter} bounds the difference in the number of properly colored triangles incident on $t$ when $Z_t=a_t$ and  when $Z_t=a_t'$, that is, the difference in the number of triangles whose pair of vertices other than $t$ is colored with $\pi(a_t) \setminus \{a_t\} $ and  that is colored with $\pi(a_t') \setminus \{a_t'\} $. As, a vertex can be present in many pairs,  proper coloring of one triangle, incident on $t$, is dependent on the porper coloring of another triangle. However, this dependency is bounded due to our assumption $\Delta_E \leq d$. Now, let us consider the following Lemma.
\begin{lem}
\label{lem:col1}
Let  $\cX$ be a set of $u$ objects \colored with $[3k]$. ${\cal F}$ be a set of $v$ pairs of objects such that an object is present in at most $d$ ($d \leq v$) many pairs and $\cP \subseteq \cX$ be a set of $w$ objects.
 ${\cal \cF}_{\{ \alpha,\beta\}} \subseteq {\cal F}$ be a set of pairs of objects that are colored with $\{\alpha,\beta\}$. $M_{\{\alpha,\beta\}}=\size{\cF_{\{\alpha,\beta\}}}$.
 $\cP_{\{\alpha,\beta\}} \subseteq \cP$ be the set of objects that are colored with $\{\alpha,\beta\}$ and $N_{\{\alpha,\beta\}}=\size{\cP_{\{\alpha,\beta\}}}$.
Then, we have
\begin{itemize}
\item[(i)]
  $\pr\left(\size{M_{\{\alpha,\beta\}} - M_{\{\alpha',\beta'\}}} \geq 8 \sqrt{dv \log u} \right)\leq \frac{4}{u^{8}}$,

\item[(ii)]
  $\E \left[\size{{M}_{\{\alpha,\beta\}} - {M}_{\{\alpha',\beta'\}}}\right] \leq \frac{\sqrt{dv}}{k}$, and

\item[(iii)]
  $\pr \left( \size{N_{\{\alpha, \beta\}} - N_{\{\alpha',\beta'\}}} \geq 4 \sqrt{w \log u} \right) \leq \frac{4}{u^8} $.
\end{itemize}
\end{lem}
\begin{proof}
%\begin{itemize}
\begin{itemize}
\item[(i)] Let $\cF=\{\{a_1,b_1\}, \ldots, \{a_v,b_v\}\}$. 
Let $X_i$ be the indicator random variable such that $X_i=1$ if and only if $\{a_i,b_i\}$ is colored with $\{\alpha,\beta\}$, where $i \in [v]$. Note that $M_{\{\alpha,\beta\}}=\sum_{i=1}^{v}X_i$. Also, $\E[X_i]=\frac{2}{9k^2}$, hence $\E[M_{\{\alpha, \beta\}}]\remove{=\E[M_{\alpha', \beta'}]}=\frac{2v}{9k^2}$. 

$X_i$ and $X_j$ are dependent if and only if $\{a_i,b_i\} \cap \{a_j,b_j\} \neq \emptyset$. As each object can be present in 
at most $d$ many pairs of objects, there are at most $2d$ many $X_j$'s on which an $X_i$ depends. Now using \emph{Chernoff-Hoeffding's type bound in the bounded dependent setting}~\cite{DubhashiP09} (see Lemma~\ref{lem:depend:high_prob} in Appendix~\ref{sec:prelim}), we have $$\pr \left(\size{M_{\{\alpha,\beta\}} - \frac{2v}{9k^2}} \geq 4\sqrt{dv\log u}\right) \leq \frac{2}{u^{8}}.$$
Similarly, one can also show that 
$\pr \left(\size{M_{\{\alpha',\beta'\}} - \frac{2v}{9k^2}} \geq 4\sqrt{dv\log u}\right) \leq \frac{2}{u^{8}}$.
Note that $$\size{M_{\{\alpha,\beta\}}-M_{\{\alpha',\beta'\}}} \leq \size{M_{\{\alpha,\beta\}} - \frac{2v}{9k^2}} + \size{M_{\{\alpha',\beta'\}} - \frac{2v}{9k^2}}.$$ Hence,
\begin{eqnarray*}
&&\pr\left(\size{M_{\{\alpha,\beta\}} - M_{\{\alpha',\beta'\}}} \geq  8 \sqrt{dv \log u} \right) \\
&&\leq \pr \left( \size{M_{\{\alpha,\beta\}} - \frac{2v}{9k^2}} + \size{M_{\{\alpha',\beta'\}} - \frac{2v}{9k^2}} \geq 8 \sqrt{dv \log u}\right)\\
&&\leq   \pr \left(\size{M_{\{\alpha,\beta\}} - \frac{2v}{9k^2}} \geq 4\sqrt{dv\log u}\right) + \pr \left(\size{M_{\{\alpha',\beta'\}} - \frac{2v}{9k^2}} \geq 4\sqrt{dv\log u}\right) \\
&&\leq {4}{u^{-8}}.
\end{eqnarray*}

\item[(ii)] Let $X_i, i \in [v]$, be the random variable such that $ X_i=1$ if $\{a_i,b_i\}$ is colored with $\{\alpha,\beta\}$; 
$X_i=-1$ if $\{a_i,b_i\}$ is colored with $\{\alpha',\beta'\}$; $X_i=0$, otherwise.
Let $X=\sum\limits_{i=1}^{v}X_i$. Note that
 $$
 M_{\{\alpha,\beta\}} - M_{\{\alpha',\beta'\}} =X= \sum\limits_{i=1}^{v}X_i.
 $$ 
So, we need to bound $\E[\size{X}] $ to prove the claim.

The random variables $X_i$ and $X_j$ are dependent if and only if $\{a_i,b_i\} \cap \{a_j,b_j\} \neq \emptyset$. As each object can be present in 
at most $d$ many pairs of objects, there are at most $2d$ many $X_j$'s on which an $X_i$ depends.
Observe that $\pr(X_i=1)=\pr(X_i=-1)=\frac{2}{9k^2}$.  So, $\E[X_i]=0$ and $\E[X_i^2]=\frac{4}{9k^2}$. If $X_i$ and $X_j$ are independent, then $\E[X_iX_j]=\E[X_i]\cdot \E[X_j]=0$. If $X_i$ and $X_j$ are dependent, then $\E[X_iX_j]  \leq  \pr(X_iX_j=1)$.
\begin{eqnarray*}
 \pr(X_iX_j=1)&=& \pr(X_i=1,X_j=1) + \pr (X_i=-1,X_j=-1)\\
&=& \pr(X_i=1)\cdot \pr (X_j=1~|~X_i=1) + \pr(X_i=-1)\cdot \pr (X_j=-1~|~X_i=-1)\\
&=& \frac{2}{9k^2}\cdot \frac{1}{3k}+ \frac{2}{9k^2}\cdot \frac{1}{3k}\\
&=& \frac{4}{27k^3}
\end{eqnarray*}
Using the expression $\E[X^2]=\sum_{i=1}^v \E[X_i^2] + 2 \cdot \sum_{1\leq i <  j \leq v} \E[X_iX_j] $
and recalling the fact that each $X_i$ depends on at most $2d$ many other $X_j$'s, we get
$$ \E[X^2] \leq v \cdot \frac{4}{9k^2} + 2dv \cdot  \frac{4}{27k^3} \leq \frac{8dv}{9k^2}.$$
Now, using $\E[\size{X}] \leq \sqrt{\E[X^2]}$, we get $\E[\size{X}]  < \frac{\sqrt{dv}}{k}$.

\item[(iii)] Let $\cP=\{o_1,\ldots,o_w\}$ be the set of $w$ objects.
 Let $X_i, i \in [w]$, be the indicator random variable such that  $ X_i=1$ if and ony if $o_i$ is colored with $\{\alpha,\beta\}$. Note that $N_{\{\alpha,\beta\}}=\sum\limits_{i=1}^{w}X_i$. Observe that $\E[X_i]=\frac{2}{3k}$ and hence, $\E \left[N_{\{\alpha,\beta\}} \right]=\frac{2w}{3k}$. Note that $X_i$ and $X_j$ are independent. 
Applying Hoeffding's inequality~(See Lemma~\ref{lem:hoeff_inq} in Appendix~\ref{sec:prelim}), we get 
$$\pr \left(\size{N_{\{\alpha,\beta\}}-\frac{2w}{3k}}\geq 2 \sqrt{w \log u} \right) \leq \frac{2}{u^8}.$$
Similarly, we can aso show that 
$\pr \left(\size{N_{\{\alpha',\beta'\}} - \frac{2w}{3k}} \geq 2\sqrt{w\log u}\right) \leq \frac{2}{u^{8}}$.
\remove{Note that $\size{N_{\{\alpha,\beta\}}-N_{\{\alpha',\beta'\}}} \leq \size{N_{\{\alpha,\beta\}} - \frac{2w}{3k}} + \size{N_{\{\alpha',\beta'\}} - \frac{2w}{3k}}$.} Hence,
\begin{eqnarray*}
&& \pr(\size{N_{\{\alpha,\beta\}} - N_{\{\alpha',\beta'\}}} \geq 4 \sqrt{w \log u})\\ 
&&\leq \pr \left( \size{N_{\{\alpha,\beta\}} - \frac{2w}{3k}} + \size{N_{\{\alpha',\beta'\}} - \frac{2w}{3k}} \geq 4 \sqrt{w \log u}\right)\\
&&\leq   \pr \left(\size{N_{\{\alpha,\beta\}} - \frac{2w}{3k}} \geq 2\sqrt{w\log u}\right) + \pr \left(\size{N_{\{\alpha',\beta'\}} - \frac{2w}{3k}} \geq 2\sqrt{w\log u}\right)  \\
&&\leq \frac{4}{u^{8}}.
\end{eqnarray*}
\remove{
\begin{eqnarray*}
&&\pr(\size{N_{\{\alpha,\beta\}} - N_{\{\alpha',\beta'\}}} \geq 4 \sqrt{w \log u})\\
&\leq & \pr \left( \size{N_{\{\alpha,\beta\}} - \frac{2w}{3k}} + \size{N_{\{\alpha',\beta'\}} - \frac{2w}{3k}} \geq 4 \sqrt{w \log u}\right)\\
&\leq&   \pr \left(\size{N_{\{\alpha,\beta\}} - \frac{2w}{3k}} \geq 2\sqrt{w\log u}\right) + \pr \left(\size{N_{\{\alpha',\beta'\}} - \frac{2w}{3k}} \geq 2\sqrt{w\log u}\right) \leq \frac{4}{u^{8}}.
\qedhere
\end{eqnarray*}
}
\end{itemize}
\end{proof}

We will now give the proof of Claim~\ref{clm:inter}.

\begin{proof}[Proof of Claim~\ref{clm:inter}]
 
\begin{itemize}
 \item[(a)] Let $S_{\ell}=\{(a_1,b_1,t),\ldots, (a_v,b_v,t)\}$. Note that $v \leq \Delta_t$. As $\Delta_E\leq d$, each vertex in $[n]$ can be present in at most d many pairs of $S_{\ell}$. 
 Now we apply Lemma~\ref{lem:col1}. Set $\cX=[n]$ and $\cF = S_{\ell}$ in Lemma ~\ref{lem:col1}.
 Observe that $N_{\ell}^{a_t}=M_{\pi(a_t) \setminus \{a_t\}}$ and $N_{\ell}^{a'_t}=M_{\pi(a'_t) \setminus \{a'_t\}}$. So, by of Lemma~\ref{lem:col1} (i),
 $$ \pr\left(\size{N_{\ell}^{a_t}-N_{\ell}^{a_t'}} \geq  8 \sqrt{d v \log n}\right) \leq  \frac{4}{n^{8}} .$$ This implies 
 $\pr\left(\size{N_{\ell}^{a_t}-N_{\ell}^{a_t'}} \geq  8 \sqrt{d \Delta_t \log n}\right) \leq  \frac{4}{n^{8}}$.
 \remove{
 \begin{eqnarray*}
 \pr\left(\size{N_{\ell}^{a_t}-N_{\ell}^{a_t'}} \geq  8 \sqrt{d v \log n}\right) &\leq&  \frac{4}{n^{8}} \\
 \pr\left(\size{N_{\ell}^{a_t}-N_{\ell}^{a_t'}} \geq  8 \sqrt{d \Delta_t \log n}\right) &\leq&  \frac{4}{n^{8}}.
 \end{eqnarray*}}

\item[(b)] Let $S_{r}=\{(t,a_1,b_1),\ldots, (t,a_v,b_v)\}$. Note that $v \leq \Delta_t$, the number of triangles incident on vertex $t$. As $\Delta_E\leq d$, each vertex in $[n]$ can be present in at most d many pairs of $S_{r}$. 
 Now we apply Lemma~\ref{lem:col1}. Set $\cX=[n]$ and $\cF = S_{r}$ in Lemma~\ref{lem:col1}.
 Observe that $N_{r}^{a_t}=M_{\pi(a_t) \setminus \{a_t\}}$ and $N_{r}^{a'_t}=M_{\pi(a'_t) \setminus \{a'_t\}}$. By Lemma~\ref{lem:col1} (ii), we get $$ \E \left[\size{N_{r}^{a_t}-N_{r}^{a_t'}}\right] \leq  \frac{\sqrt{dv}}{k}  \leq \frac{\sqrt{d\Delta_t}}{k} .$$
\remove{ \begin{eqnarray*}
 \E \left[\size{N_{r}^{a_t}-N_{r}^{a_t'}}\right] &\leq&  \frac{\sqrt{dv}}{k}  \leq \frac{\sqrt{d\Delta_t}}{k} .
 \end{eqnarray*}}
 
\item[(c)]
Let $S_{\ell r}=\{(a_1,t,b_1),\ldots,(a_w,t,b_w)\}$. Without loss of generality, assume that $a_i \in [t-1]$ and $b_i \in [n] \setminus [t]$.
\remove{ $1 \leq a_i \leq t-1$ and $t+1 \leq b_i \leq n$, for each $i \in [w]$.} Note that $w \leq \Delta_t$.
   Given that the vertex $t$ is colored with color $c$ and we know $Z_1,\ldots,Z_{t-1}$, define the set $P_c$ as 
   $$P_{c} := \{(a,t,b) \in S_{\ell r}~: \mbox{$t$ is colored with $c$ and $\pr((a,t,b)$ is properly colored}) > 0\}.$$ 
   
   Let $Q_c=\size{P_c}$.  Observe that for $(a,t,b) \in S_{\ell r}$, $\pr((a,t,b)~\mbox{is properly colored}) > 0$ if and only if $a$ is colored with some color in $\pi(c) \setminus \{c\}$. Now we apply Lemma~\ref{lem:col1}. Set $\cX=[n]$, $\cP=\{a_1,\ldots,a_w\}$. Observe that 
   $\cP_{\pi(a_t) \setminus {a_t}}=P_{a_t} $ and   $\cP_{\pi(a'_t) \setminus {a'_t}}=P_{a'_t}$. By (iii) of Lemma~\ref{lem:col1}, 
   $$\pr\left(\size{Q_{a_t} - Q_{a'_t}} \geq 4\sqrt{w \log n}\right) \leq \frac{4}{n^8}.$$
   Let $\cE$ be the event that $ \size{Q_{a_t} - Q_{a'_t}} \geq 4\sqrt{w \log n}$. So, $\pr(\cE) \leq \frac{4}{n^8}$. Assume that $\cE$ has not occurred. Let $P=P_{a_t} \cap P_{a'_t}=\{(x_1,t,y_1),\ldots,(x_{q},t,y_{q})\}$. Note that $q \leq w \leq \Delta_{t}$. Recall that $Z_x$ is the random variable that denotes the color assigned to vertex $x \in [n]$.  Let $X_i, i \in [q]$, be the random variable such that $X_i=1$ if $y_i$ is colored 
  with $\pi(a_t) \setminus \{ Z_{x_i},a_t \}$; $X_i=-1$ if $y_i$ is colored with $\pi(a'_t) \setminus \{Z_{x_i},a'_t\}$; $X_{i} = 0$, otherwise. Let $X = \sum_{i=1}^q X_i$. Observe that $X_i$ and $X_j$ are dependent if and only if $y_i=y_j$. As $\Delta_E \leq d$, there can be at most $d$ many $y_j$'s such that $y_i=y_j$. So, an $X_i$ depends on at most $d$ many other $X_j$'s.  
   
   Observe that $\pr(X_i=1)=\pr(X_i=-1)=\frac{1}{3k}$. So, $\E[X_i=0]$ and $\E[X_i^2]=\frac{2}{3k}$.
   If $X_i$ and $X_j$ are independent, then $\E[X_iX_j]=0$.
   If $X_i$ and $X_j$ are dependent, then 
   $$
   \E[X_iX_j] \leq  \pr(X_i=1,X_j=1) + \pr(X_i=-1,X_j=-1)
   \leq \pr(X_i=1) + \pr (X_j=-1) = \frac{2}{3k}.
   $$

Using the expression 
$\E[X^2]=\sum_{i=1}^v \E[X_i^2] + 2 \cdot \sum_{1\leq i <  j \leq v} \E[X_iX_j]$ and the fact that each $X_i$ depends on at most $d$ many other $X_j$'s, we get
$$
\E[X^2] \leq v \cdot \frac{2}{3k} + d v \cdot \frac{2}{3k} \leq \frac{dv}{k} \leq \frac{d \Delta_t}{k}.
$$
Since,  $\E[\size{X}] \leq \sqrt{\E[X^2]}$, we get
$ \E[\size{X}] \leq   \sqrt{\frac{d\Delta_t}{k}}$. Using $\Delta_E \leq d$, we have  
\begin{eqnarray*}
  \E [\mid N_{\ell r}^{a_t}-N_{\ell r}^{a_t'}\mid~|~\cE^c] &=& d \cdot \mid Q_{a_t} -Q_{a'_t} \mid + \E[\size{X}] \\
  &<& 4d \sqrt{\Delta_t \log n} +  \sqrt{\frac{d\Delta_t}{k}} < 5d\sqrt{\Delta_t \log n}.
\end{eqnarray*}

Observe that $  \E [\mid N_{\ell r}^{a_t}-N_{\ell r}^{a_t'} \mid~|~\cE ] \leq w \leq \Delta_t $.
Putting everything together,
\begin{eqnarray*}
\E \left[ \mid {N_{\ell r}^{a_t}-N_{\ell r}^{a_t'}} \mid \right] &=& \pr(\cE) \cdot \E \left[\mid{N_{\ell r}^{a_t}-N_{\ell r}^{a_t'}}\mid \; | \; \cE \right] + \pr(\cE^c) \cdot \E \left[\mid{N_{\ell r}^{a_t}-N_{\ell r}^{a_t'}}\mid \; | \; \cE^c \right]\\
&<& \frac{4}{n^8} \cdot \Delta_t + 1 \cdot 5d\sqrt{\Delta_t \log n} \leq 6d\sqrt{\Delta_t \log n}
\end{eqnarray*}
\end{itemize}
\end{proof}

%% file: estimate.tex
%\section{Exact estimation and Threshold-Approx-Estimation \complain{(Why is the title like this?)}}
\section{Estimation: Exact and Approximate}
\label{sec:coarse}
%\comments{Gopi: This Section is completely different from the submitted version.}
\noindent
In this Section, we prove Lemmas~\ref{lem:exact_decide} (restated as Lemma~\ref{lem:exact_decide-restated}),~\ref{lem:decide_exact} (restated as Lemma~\ref{lem:decide_exact-restated}) and ~\ref{lem:exact} (restated as Lemma~\ref{lem:exact-restated}). We first prove Lemmas~\ref{lem:decide_exact} and ~\ref{lem:exact}, whose proofs are very similar. Then we prove  Lemma~\ref{lem:exact_decide} that in turn uses Lemma~\ref{lem:decide_exact}.

\begin{lem}[Lemma~\ref{lem:exact} restated]
There exists a deterministic algorithm that given any disjoint subsets $A,B,C \subset V(G)$ of any graph $G$, can determine the exact value of $t(A,B,C)$ using $\Oh(t(A,B,C) \log n)$ \tis queries. 
\label{lem:exact-restated}
\end{lem}
\begin{proof}
We initialize a tree ${\cal T}$ with $(A,B,C)$ as the root. We build the tree such that each node is labeled with either $0$ or $1$. If $t(A,B,C)=0$, we label the root with $0$ and terminate. Otherwise, we label the root with $1$ and do the following as long as there is a leaf node $(U,V,W)$ labeled with $1$.
\begin{itemize}
\item[(i)] If $t(U,V,W)=0$, then we label $(U,V,W)$ with $0$ and go to other leaf node labeled as $1$ if any. Otherwise, we label $(U,V,W)$ as $1$ and do the following.
\item[(ii)] If $\size{U}=\size{V}=\size{W}=1$, then we add one node $(U,V,W)$ as a child of $(U,V,W)$ and label the new node as $0$. Then we go to other leaf node labeled as $1$ if any.
\item[(iii)] If $\size{U}=1, \size{V}=1$ and $\size{W} >1$, then we partition the set $W$ into $W_1$ and $W_2$ such that
 $\size{W_1}=\lceil \frac{\size{W}}{2} \rceil$ and $\size{W_2}=\lfloor \frac{\size{W}}{2} \rfloor$ ; and we add $(U,V,W_1)$ and $(U,V,W_2)$ as two children of $(U,V,W)$.
 The case $\size{U}=1, \size{V}>1$, $\size{W} =1$ and $\size{U}>1, \size{V}=1$, $\size{W} =1$ are handled similarly.
 \item[(iv)] If $\size{U}=1, \size{V}>1$ and $\size{W} >1$, then we partition the set $V$ into $V_1$ and $V_2$ (similarly, $W$ into $W_1$ and $W_2$) such that
  $\size{V_1}=\lceil \frac{\size{V}}{2} \rceil$ and $\size{V_2}=\lfloor \frac{\size{V}}{2} \rfloor$  ($\size{W_1}=\lceil \frac{\size{W}}{2} \rceil$ and $\size{W_2}=\lfloor \frac{\size{W}}{2} \rfloor$); and we add $(U,V_1,W_1)$, $(U,V_1,W_2)$, $(U,V_2,W_1)$ and $(U,V_2,W_2)$ as four children of $(U,V,W)$.
 The case $\size{U}>1, \size{V}>1$, $\size{W} =1$ and $\size{U}>1, \size{V}=1$ $\size{W} >1$ are handled similarly.
 
  \item[(v)] If $\size{U}>1, \size{V}>1$ and $\size{W} >1$, then we partition the sets $U,V,W$ into $U_1$ and $U_2$; $V_1$ and $V_2$; $W_1$ and $W_2$, respectively, such that
   $\size{U_1}=\lceil \frac{\size{U}}{2} \rceil$ and $\size{U_2}=\lfloor \frac{\size{U}}{2} \rfloor$; $\size{V_1}=\lceil \frac{\size{V}}{2} \rceil$ and $\size{V_2}=\lfloor \frac{\size{V}}{2} \rfloor$;  $\size{W_1}=\lceil \frac{\size{W}}{2} \rceil$ and $\size{W_2}=\lfloor \frac{\size{W}}{2} \rfloor$. We add $(U_1,V_1,W_1)$, $(U_1,V_1,W_2)$, $(U_1,V_2,W_1)$, $(U_1,V_2,W_2)$ $(U_2,V_1,W_1)$, $(U_2,V_1,W_2)$, $(U_2,V_2,W_1)$ and $(U_2,V_2,W_2)$ as eight children of $(U,V,W)$.

\end{itemize}
Let $\cT'$ be the tree after deleting all the leaf nodes in $\cT$. Observe that $t(A,B,C)$ is the number of  leaf nodes in $\cT'$; and
\begin{itemize}
\item the height of ${\cal T}$ is bounded by $\max \{\log \size{A}, \log \size{B}, \log \size{C} \} + 1 \leq 2\log n $,
\item the query complexity of the above procedure is bounded by the number of nodes in ${\cal T}$ as we make at most one query per node of $\cT$.
%\item $t(U,V,W)$ is bounded by the number of leaf nodes and the number of leaf nodes is bounded by $8 t(U,V,W)$.
\end{itemize}
The number of nodes in $\cT'$, the number of internal nodes of $\cT$, is bounded by $2t(A,B,C) \log n$. So,  the number of leaf nodes in $\cT$ is at most $16t(A,B,C) \log n$ and hence the total number of nodes in $\cT$ is at most $16t(U,V,W) \log n$. Putting everything together, the required query complexity is $\Oh(t(A,B,C)\log n)$. 
\end{proof}

\begin{lem}[Lemma~\ref{lem:decide_exact} restated]
There exists a deterministic algorithm that given any disjoint subsets $A,B,C \subset V(G)$ of any graph $G$ and a threshold parameter $\tau \in \\cN$, can decide whether $t(A,B,C) \leq \tau$ using $\Oh(\tau \log n)$ \tis queries. 
\label{lem:decide_exact-restated}
\end{lem}
\begin{proof}
The algorithm proceeds similar to the one presented in the Proof of Lemma~\ref{lem:exact} by initializing a tree $\cT$ with $(A,B,C)$ as the root. If 
$t(A,B,C) \leq \tau$, then we can find $t(A,B,C)$ by using $16t(A,B,C)\log n$ many queries and the number 
of nodes in $\cT$ is bounded by $16t(A,B,C)\log n$. So, if the number of nodes in $\cT$ is more than
$16\tau \log n$ at any instance during the execution of the algorithm, we report $t(G) > \tau$ and terminate. Hence, the query complexity is bounded by the number of nodes in $\cT$, which is $\Oh(\tau \log n)$.
\end{proof}

\begin{algorithm}[h]
%\scriptsize

%\SetAlgoLined

\caption{{{\sc Threshold-Approx-Estimate}($G,\tau, \eps$)}}
\label{algo:exact}
\KwIn{A parameter $\tau$ and an $\eps \in (0,1)$.}
\KwOut{Either report $t(G)>\tau$ or find an $(1 \pm \eps)$-approximation of ${t}(G)$.}
%\LinesNumbered
\For{($i=1$ to $\cN=\frac{18 \log n}{\eps^2}$)}
{
 Partition $V(G)$ into three parts such that each vertex is present in one of $A_i,B_i,C_i$ with probability $1/3$ independent of the other vertices.\\
Run the algorithm corresponding to Lemma~\ref{lem:decide_exact} to determine if $t(A_i,B_i,C_i)>\tau$. If yes, we report $t(G) > \tau$ and {\sc Quit}. Otherwise, we have the exact vale of $t(A_i,B_i,C_i)$
}
Report $\hat{t}=\frac{9 \sum \limits_{i}^{\cN}t(A_i,B_i,C_i)}{2 \cN}$ as the output.
\end{algorithm}

{\begin{lem}[Lemma~\ref{lem:exact_decide} restated]
 {There exists an algorithm that for any graph $G$, a threshold parameter $\tau \in \N$ and an $\eps \in (0,1)$, determines whether $t(G) > \ \tau$. If $t(G) \leq \ \tau$, the algorithm gives a $(1 \pm \eps)$-approximation to $t(G)$ by using $\cO(\frac{\tau \log^2 n}{\eps^2})$ many \tis queries with probability at least $1- n^{-10}$.}
 \label{lem:exact_decide-restated}
\end{lem}}

\begin{proof}

We show that Algorithm~\ref{algo:exact} satisfies the given condition in the statement of Lemma~\ref{lem:decide_exact}. Note that
 {\sc Threshold-Approx-Estimate} calls the algorithm corresponding to Lemma~\ref{lem:decide_exact} at most $\cN=\frac{18 \log n}{\eps^2}$ times, where each call can be executed by $\Oh(\tau \log n)$ \tis queries. So, the total query complexity of {\sc Threshold-Approx-Estimate} is $\Oh(\cN \cdot \tau \log n)=\Oh\left( \frac{\tau \log ^2 n}{\eps^2}\right)$.

Now, we show the correctness of {\sc Threshold-Approx-Estimate}. If there exists an 
$i \in [\cN]$, such that $t(A_i,B_i,C_i)>\tau$, then we report $t(G)>\tau$ and {\sc Quit}. Otherwise, by Lemma~\ref{lem:decide_exact}, we have the exact values of $t(A_i,B_i,C_i)$'s. We will be done by showing that $\hat{t}$ is an $(1 \pm \eps)$-approximation to $t(G)$ with probability at least $1-n^{-10}$.
From the description of the algorithm, each triangle in $G$ will be counted in $t(A_i,B_i,C_i)$ with probability $\frac{2}{9}$. We have $\E[t(A_i,B_i,C_i)]=\frac{2}{9}~t(G)$, and the expectation of the sum and the estimate $\hat{t}$ is  
$$
%\E[t(A_i,B_i,C_i)]=\frac{2}{9}t(G), 
\E\left[\sum\limits_{i=1}^\cN t(A_i,B_i,C_i)\right]=\frac{2}{9}\cN \cdot t(G) ~~ \mbox{ and } ~~ \E[\hat{t}]=t(G).
$$
Therefore, we have
\begin{eqnarray*}
\pr \left( \size{\hat{t}-t(G)} \geq  \eps\cdot  t(G) \right) = 
\pr \left( \size{\sum\limits_{i=1}^\cN t(A_i,B_i,C_i)-\frac{2}{9}\cN \cdot t(G)} \geq \frac{2 \eps}{9} \cN \cdot t(G)  \right)
\end{eqnarray*}
To bound the above probability, we apply Hoeffding's inequality~(See Lemma~\ref{lem:hoeff_inq} in Appendix~\ref{sec:prelim}) along with the fact that $0 \leq t(A_i,B_i,C_i) \leq \tau$ for all $i \in [\cN] $, and we get 
\begin{eqnarray*}
 \pr \left( \size{\hat{t}-t(G)} \geq  \eps\cdot t(G) \right) \leq \frac{1}{n^{10}}.
\end{eqnarray*}
\end{proof}

%% file: coarse.tex
\section{Coarse estimation}
\label{sec:coarse_new}
\noindent
\remove{In the intermediate step of the algorithm, it might happen that we have \emph{large} number $t(A,B,C)$'s  such that the value of $t(A,B,C)$'s are also large. In this case, we find a $\cO(\log^ 2 n)$-multiplicatice approximate estimate of $t(A,B,C)$ and we apply Lemma~\ref{lem:importance}. Given $A,B,C \subset V(G)$, Algorithm~\ref{algo:coarse} finds the desired estimate of $t(A,B,C)$ using $\cO(\log^ 4 n)$ \tis queries. Lemma~\ref{lem:coarse_main} states the result, where as .}
\noindent We now prove Lemma~\ref{lem:coarse_main}. Algorithm~\ref{algo:coarse} corresponds to Lemma~\ref{lem:coarse_main}. Algorithm~\ref{algo:verify} is a subroutine in Algorithm~\ref{algo:coarse}. Algorithm~\ref{algo:verify} determines whether a given estimate $\hat{t}$ is correct upto a $\Oh(\log ^2 n)$ factor.\remove{  Given an estimate $\hat{t}$ of $t(A,B,C)$, whether it is correct upto $\Oh(\log ^2 n)$
factor is decided by Algorithm~\ref{algo:verify}, which is called as a subroutine in Algorithm~
\ref{algo:coarse}.} Lemmas~\ref{lem:coarse} and~\ref{lem:coarse1} are intermediate results needed to prove Lemma~\ref{lem:coarse_main}.
\begin{algorithm}
%\scriptsize
\caption{\verest($A,B,C,\hat{t}$)}\label{algo:verify}
\KwIn{Three pairwise disjoint set $A,B,C \subseteq V(G)$ and $\hat{t}$.}
\KwOut{If $\hat{t}$ is a good estimate, then {\sc Accept}. Otherwise, {\sc Reject}.}
%\LinesNumbered
\Begin
{	
\For{($i= 2\log n$ to $0$)}{
\For{($j=\log n$ to $0$)}
{
Find $A_{ij}\subseteq A$, $B_{ij} \subseteq B$, $C_{ij} \subseteq C$ by sampling each element  of $A$, $B$ and $C$, respectively with probability $\min\{\frac{2^i}{\hat{t}},1\}$, $\min\{\frac{2^j}{2^i} \log n,1\}$, $\frac{1}{2^j}$, respectively.\\
\If{$(t(A_{ij},B_{ij},C_{ij}) \neq 0)$}
 {\sc Accept}
}
}
 {\sc Reject}
}
\end{algorithm}
\begin{lem}
\label{lem:coarse}
If $\hat{t} \geq 64 t(A,B,C) \log^3 n$, $\pr( \mbox{\verest ($A,B,C,\hat{t}$) accepts}) \leq\frac{1}{20} $.
\end{lem}
\begin{proof}
Let $T(A,B,C)$ denote the set of triangles having vertices $a \in A$, $b \in B$ and $c \in C$, where $A, B$ and $C$ are disjoint subsets of $V(G)$. For $(a,b,c) \in T(A,B,C)$ such that $a \in A, b\in B,c \in C$, let $X^{ij}_{(a,b,c)}$ denote the indicator random variable such that $X^{ij}_{(a,b,c)}=1$ if and only if $(a,b,c) \in T(A_{ij},B_{ij},C_{ij})$ and $X_{ij}=\sum\limits_{(a,b,c) \in T(A,B,C)} X^{ij}_{(a,b,c)}$. Note that $t(A_{ij},B_{ij},C_{ij})=X_{ij}$. $(a,b,c)$ is present in $T(A_{ij},B_{ij},C_{ij})$ if $a \in A_{ij}$, $b \in B_{ij}$ and $c \in C_{ij}$. So, 
 $$
 \pr\left( X^{ij}_{(a,b,c)} =1\right) \leq \frac{2^i}{\hat{t}}\cdot \frac{2^j}{2^i} \log n \cdot \frac{1}{2^j}=\frac{\log n}{\hat{t}}~\mbox{and}~\E\left[X_{ij}\right] \leq \frac{t(A,B,C)}{\hat{t}} \log n .
 $$ 
As $X_{ij} \geq 0$,
$$
\pr \left( X _{ij}\neq 0\right) = \pr(X_{ij} \geq 1)  \leq \E \left[X_{ij} \right] \leq \frac{t(A,B,C)}{\hat{t}} \log n.
$$
 Now using the fact that $\hat{t} \geq 64t(A,B,C) \log^3 n$, we have
$ \pr\left(X_{ij} \neq 0 \right) \leq \frac{1}{64 \log^2  n} .$ Observe that \verest accepts if and only if there exists $i,j \in \{0,\ldots,\log n\}$ such that $X_{ij} \neq 0$. Using the union bound, we get
\begin{eqnarray*}
\pr \left( \mbox{\verest accepts} \right) &\leq&  \sum\limits_{0 \leq i \leq
                                          2 \log n}~\sum\limits_{0
                                          \leq j \leq  \log
                                          n}\pr\left( X_{ij} \neq 0 \right) \\
&\leq& \frac{(2\log n +1)(\log n +1)}{32 \log^2 n} \\
&\leq& \frac{1}{20}.
\end{eqnarray*}
\end{proof}
\begin{lem}
\label{lem:coarse1}
If $\hat{t} \leq \frac{t(A,B,C)} {32 \log n}$, $ \pr( \mbox{\verest ($A,B,C,\hat{t}$) accepts}) \geq \frac{1}{5}$ .
\end{lem}
\begin{proof}
 For $p \in \{0,\ldots, 2\log n \}$, let $A^{p} \subseteq A$ be the set of vertices such that for each $a \in A_p$, the number of triangles of the form $(a,b,c)$ with $(b,c) \in B \times C$ , lies between $2^{p}$ and  $ 2^{p+1}-1$.

  For $a \in A^p$ and $q \in \{0,\ldots, \log n \}$, let $ B^{pq}(a) \subseteq B$ is the set of vertices such that for each $b \in B$, the number of triangles of the form $(a,b,c)$ with $c\in C$ lies between $2^{q}$ and  $ 2^{q+1}-1$
 \remove{$$ C^{pq}=\{c \in C : \exists ~ a \in A~\mbox{ and}~ b \in B~\mbox{such that}~(a,b,c)~\mbox{ is a traingle and}~ 2^{p} \leq  \Delta _a< 2^{p+1}~\mbox{and}~ 2^q \leq \Delta_{(a,b)} < 2^{q+1} \}.$$}
We need the following Claim to proceed further.
\begin{cl}
\label{clm:verify}
\begin{itemize}
\item[(i)] There exists $p \in \{0,\ldots, 2\log n  \}$ such that $\size{A^p} > \frac{t(A,B,C)}{2^{p+1}(2\log n +1)}$.
\item[(ii)] For each $a \in A^{p}$, there exists $q \in \{0,\ldots, \log n\}$ such that $\size{B^{p q}(a)} > \frac{2^{p}}{2^{q+1}( \log n +1)}$.
\end{itemize}
\end{cl}
%\vspace{-0.2in}
\begin{proof}
\begin{itemize}
\item[(i)] Observe that $t(A,B,C)=\sum_{p=0}^{2\log n } t(A^{p},B,C)$ as the sum takes into account all incidences of vertices in $A$. So, 
there exists $p \in \{0, \ldots, 2\log n\}$ such that $t(A^{p},B,C) \geq \frac{t(A,B,C)}{2\log n + 1}$. From the definition of $A^{p}$, $t(A^{p},B,C) < \size{A^{p}} \cdot 2^{p +1}$. Hence, there exists $p \in \{0, \ldots, 2\log n\}$ such that 
$$ \size{A^{p}} > \frac{t(A^{p},B,C) }{2^{p+1}} \geq  \frac{t(A,B,C)}{2^{p+1} (2\log n+1)}.$$
\item[(ii)] Observe that $\sum_{q=0}^{\log n} t(\{a\},B^{p q}(a),C)=t(\{a\},B,C)$. So, 
there exists $q \in \{0, \ldots, \log n\}$ such that $t(\{a\},B^{p q}(a),C) \geq \frac{t(\{a\},B,C)}{ \log n +1}$. From the definition of $B^{p q}(a)$, $t(\{a\},B^{p q}(a),C) < \size{B^{p q}(a)} \cdot 2^{q +1}$. Hence, there exists $q \in \{0, \ldots, \log n\}$ such that 
$$ \size{B^{p q}(a)} > \frac{t(\{a\},B^{p q}(a),C) }{2^{q+1}} \geq  \frac{t(\{a\},B,C)}{2^{q+1}( \log n+1)} \geq \frac{2^{p}}{2^{q+1} (\log n +1)}.$$

\end{itemize}
\end{proof}
We come back to the proof of Lemma~\ref{lem:coarse1}. We will show that \verest accepts with probability at least $\frac{1}{5}$ when loop executes for $i=p$, where $p$ is such that $\size{A^p} > \frac{t(A,B,C)}{2^{p+1}(2\log n +1)}$. The existence of such a $p$ is evident from  Claim~\ref{clm:verify} (i). 

Recall that $A_{p q}\subseteq A, B_{p q}\subseteq B$ and $C_{p q} \subseteq C$ are the samples obtained 
when the loop variables $i$ and $j$ in Algorithm~\ref{algo:verify} attain values $p$ and $q$, respectively. Observe that
$$\pr\left(A_{p q} \cap A^{p}  = \emptyset \right) \leq  \left( 1- \frac{2^{p}}
{\hat{t}}\right)^{\size{A^{p}}} \leq e^{-\frac{2^{p}}{\hat{t}}\size{A^{p}}} \leq e^{-\frac{2^{p}}{\hat{t}}\frac{t(A,B,C)}{2^{p+1} \log n}}=e^{-\frac{t(A,B,C)}{2\hat{t}(2\log n +1)}}.$$
Now using the fact that $\hat{t} \leq \frac{t(A,B,C)} {32 \log n}$ and $n \geq 64$,
$$\pr\left(A_{p q} \cap A^{p}  = \emptyset \right) \leq \frac{1}{e^6} .$$

Assume that $A_{p q} \cap A^{p}  \neq \emptyset$ and $a \in A_{p q} \cap A^{p} $. By Claim~\ref{clm:verify}~(ii), there exists $q \in \{0,\ldots, \log n\}$, such that $B^{p q}(a) \geq \frac{2^{p}}{2^{q + 1 } (\log n +1)}$. {Note that $q$ depends on $a$}. Observe that we will be done, if we can show that \verest accepts when loop executes for $i=p$ and $j= q$. Now,
$$\pr \left( B_{p q} \cap B^{p q}(a)  = \emptyset ~|~ A_{p q} \cap A^{p}  \neq \emptyset\right) \leq  \left( 1-\frac{2^{q}}{2^{p} }\log n\right)^{\size{B^{p q}(a)}} \leq \frac{1}{e^{3/7}}.$$

Assume that $A_{p q} \cap A^{p}  \neq \emptyset$, $ B_{p q} \cap B^{p q}(a)  \neq \emptyset$ and $b \in B_{p q} \cap B^{p q}(a)$. Let $S$ be the set such that 
$(a,b,s)$ is a triangle in $G$ for each $s \in S$. Note that $\size{S}\geq 2^{q}$. So,
$$\pr \left( C_{p q} \cap S  = \emptyset ~|~ A_{p q} \cap A^{p} \neq \emptyset ~\mbox{and}~B_{p q} \cap B^{p q}(a)  \neq \emptyset \right) \leq \left( 1- \frac{1}{2^{q}} \right)^{2^{q}} \leq \frac{1}{e}.$$
Observe that \verest accepts if $ t(A_{p q},B_{p q},C_{p q}) \neq 0$. Also,
$ t(A_{p q},B_{p q},C_{p q}) \neq 0$ if $A_{p q} \cap 
A^{p}  \neq \emptyset$, $B_{p q} \cap B^{p q}(a)  \neq \emptyset$ and $C_{p 
q} \cap S \neq \emptyset$. Hence, 
\begin{eqnarray*}
\pr(\mbox{\verest accepts}) &\geq& \pr \left( \mbox{$A_{p q} \cap A^{p}  \neq \emptyset$, $B_{p q} \cap B^{p q}(a)  \neq \emptyset$ and $C_{p q} \cap S \neq \emptyset$} \right) \\
&=& \pr\left( A_{p q} \cap A^{p}  \neq \emptyset \right) \cdot \pr \left( B_{p q} \cap B^{p q}(a)  \neq \emptyset ~|~ A_{p q} \cap A^{p}  \neq \emptyset\right)\\
&& \cdot 
\pr \left( C_{p q} \cap S  \neq \emptyset ~|~ A_{p q} \cap A^{p} \neq \emptyset ~\mbox{and}~B_{p q} \cap B^{p q}(a)  \neq \emptyset \right)\\
&>& \left(1-\frac{1}{e^6} \right)\left(1-\frac{1}{{e}^{3/7}}\right)\left( 1-\frac{1}{e}\right) > \frac{1}{5}.
\end{eqnarray*}
\end{proof}

\remove{\begin{rem}
The query complexity of \cest is $\cO( \log ^4 n)$.
\end{rem}}

\remove{\begin{lem}
\label{lem:coarse_main}
\cest($A,B,C$) outputs an estimate $\hat{t}$ with probability at least $1-\frac{1}{n^4}$ such that $\frac{t(A,B,C)}{32 \log n} \leq \hat{t} \leq 16 t(A,B,C) \log^3 n $.
\end{lem}}
\remove{\begin{proof}[Proof sketch of Lemma~\ref{lem:coarse_main}]
For a particular $\hat{t}$, 
let $X_i$ be the indicator random variable such that $X_i=1$ if and only if the  $i^{th}$ execution of \verest 
accepts. Also take $X=\sum_{i=1}^\Gamma X_i$. \cest gives output $\hat{t}$  if $X > \frac{\Gamma}{10}$.
Applying Chernoff-Hoeffding's inequality~(See (i) of Lemma~\ref{lem:cher_bound} along with union bound, we can show that \cest gives some $\hat{t}$ as output with probability at least $1-\frac{1}{n^9}$ satisfying $ \frac{t(A,B,C)}{64 \log n} \leq \hat{t}\leq 64  t(A,B,C) \log^3 n.$ For the detailed proof See Apendix~\ref{app:proof}.
\end{proof}}
\begin{algorithm}
%\scriptsize
\caption{\cest($A,B,C$)}
\label{algo:coarse}
\KwIn{Three pairwise disjoint sets $A,B,C \subset V(G)$.}
\KwOut{An estimate $\hat{t}$ for $t(A,B,C)$.}
%\LinesNumbered
\Begin
	{
\For{$(~\hat{t}= n^3,n^3/2,\ldots, 1)$}{
Repeat \verest $(A,B,C,\hat{t})$ for $\Gamma=2000 \log n $ times.
If at least $\frac{\Gamma}{10}$ many \verest accepts, then output $\widetilde{t}=\frac{\hat{t}}{\log n}$.
}}
\end{algorithm}
 \begin{lem}[Lemma~\ref{lem:coarse_main} restated]
There exists an algorithm that given disjoint subsets $A,B,C \subset V(G)$ of any graph $G$, returns an estimate $\widetilde{t}$ satisfying 
$$
	\frac{t(A,B,C)}{64 \log^2 n} \leq \widetilde{t} \leq 64 t(A,B,C) \log^2 n
$$
with probability at least $1-n^{-9}$. Moreover, the query 
complexity of the algorithm is $\Oh(\log ^4 n)$.
\end{lem} 
\begin{proof}
Note that an execution of \cest for a particular $\hat{t}$, repeats \verest for $\Gamma =2000 \log n$ times 
and gives output $\hat{t}$ if at least $\frac{\Gamma}{10}$ many \verest accepts. For a particular $\hat{t}$, 
let $X_i$ be the indicator random variable such that $X_i=1$ if and only if the  $i^{th}$ execution of \verest 
accepts. Also take $X=\sum_{i=1}^\Gamma X_i$. \cest gives output $\hat{t}$  if $X > \frac{\Gamma}{10}$.

Consider the execution of \cest for a particular $\hat{t}$. If $\hat{t}  \geq 32 t(A,B,C) \log ^3 n$, we first show that \cest accepts with probability at least $1-\frac{1}{n^5}$. Recall Lemma~\ref{lem:coarse}. If $\hat{t} \geq 64t(A,B,C) \log ^3 n$, $\pr(X_i =1) \leq \frac{1}{20}$ and hence $\E[X] \leq \frac{\Gamma}{20}$. By using Chernoff-Hoeffding's inequality~(See Lemma~\ref{lem:cher_bound}~(i) in Appendix~\ref{sec:prelim}), 
$$\pr \left(X > \frac{\Gamma}{10} \right)=\pr\left( X > \frac{\Gamma}{20} + \frac{\Gamma}{20}\right) \leq \frac{1}{n^{10}}.$$ By using the union bound for all $\hat{t}$, the probability that \cest outputs some $\widetilde{t}=\frac{\hat{t}}{\log n} \geq 16t(A,B,C) \log ^2 n$, is at most $\frac{3 \log n}{n^{10}}$. 

Now consider the instance when the for loop in \cest executes for a $\hat{t}$ such that $\hat{t} \leq \frac{t(A,B,C)}{ 32 \log  n}$. In this situation, $\pr(X_i=1) \geq \frac{1}{5}$. So, $\E[X] \geq \frac{\Gamma}{5}$. 
By using Chernoff-Hoeffding's inequality~(Lemma~\ref{lem:cher_bound}~(ii) in Appendix~\ref{sec:prelim}), 
$$\pr\left(X \leq \frac{\Gamma}{10} \right) \leq \pr \left( X < \frac{3\Gamma}{20}\right) = \pr\left(X < \frac{\Gamma}{5} -\frac{\Gamma}{20} \right) \leq \frac{1}{{n^{10}}}.$$ 
 By using the union bound for all $\hat{t}$, the probability that \cest outputs some $\widetilde{t}=\frac{\hat{t}}{\log n} \leq \frac{t(A,B,C)}{ 16 \log ^2  n}$, is at most $\frac{3 \log n}{n^{10}}$.  

Observe that, \cest gives output $\widetilde{t}$ that satisfies either $\widetilde{t}\geq 64 t(A,B,C)\log ^2 n$ or $\widetilde{t} \leq \frac{t(A,B,C)}{32 \log ^2 n}$ is at most $\frac{3\log n}{n^{10}} +\frac{3\log n}{n^{10}} \leq \frac{1}{n^9}$.

Putting everything together, \cest gives some $\widetilde{t}$ as output with probability at least $1-\frac{1}{n^9}$ satisfying $$ \frac{t(A,B,C)}{64 \log^2 n} \leq \widetilde{t}\leq 64  t(A,B,C) \log^2 n.$$

From the description of \verest and \cest, the query complexity of \verest is $\Oh(\log ^2 n)$ and \cest calls \verest $\Oh(\log^2 n)$ times. Hence, \cest makes $\Oh(\log ^4 n)$ many queries. 
\end{proof}

%% file: finalalgo.tex
\section{The final triangle estimation algorithm: Proof of Theorem 1}
\label{sec:finaltrialgo-app}
%\color{blue}
\noindent {Now we design an algorithm for $(1 \pm \eps)$-multiplicative approximation of $t(G)$. If $\eps \leq \frac{\sqrt{d}\log ^{9/2} n}{n^{3/4}} $, we query for $t(\{a\},\{b\},\{c\})$ for all distinct $a,b,c \in V(G)$ and compute the exact value of $t(G)$. So, we assume that $\eps > \frac{\sqrt{d}\log ^{9/2}n}{n^{3/4}}$.}

{We build a data structure such that it maintains two things at any point of time.}

\begin{description}
\item[(i)] An accumulator $\psi$ for the number of triangles. We initialize $\psi =0$.

\item[(ii)] A set of tuples $(A_1,B_1,C_1,w_1),\ldots, (A_{\zeta},B_{\zeta},C_{\zeta},w_{\zeta})$, where tuple $(A_i,B_i,C_i)$ corresponds to the tripartite 
subgraph $G(A_i,B_i,C_i)$ and $w_i$ is the weight associated to $G(A_i,B_i,C_i)$. Initially, there is no tuple in our data structure. 
\end{description}
Before discussing the steps of our algorithm, some remarks about our sparsification lemmas (Lemmas~\ref{theo:sparse} and~\ref{theo:sparse1}) are in order.
\begin{rem}\label{rem:gen-spar}
\begin{itemize}
\item[(i)] In Lemma~\ref{theo:sparse}, $\frac{9k^2}{2}\sum\limits_{i=1}^kt(V_i,V_{k+i},V_{2k+i})$ is an $(1\pm \lambda)$-approximation of $t(G)$ when 
$$\kappa_1 d k^2 \sqrt{t(G) } \log n \leq \lambda ~ t(G) \Leftrightarrow t(G) \geq \frac{\kappa_1^2d^2k^4 \log ^2 n}{\lambda^2}.$$
In our algorithm, we apply Lemma~\ref{theo:sparse} for $k=1$. Also, we require $\lambda=\frac{\eps}{6 \log n}$. So, Lemma~\ref{theo:sparse} gives useful result in our algorithm when $t(G) \geq \frac{36\kappa_1^2 d^2 \log^4 n}{\eps^2}$. 
\item[(ii)] In Lemma~\ref{theo:sparse1}, ${k^2}\sum\limits_{i=1}^k t(V_i,V_{k+i},V_{2k+i})$ is an $(1\pm \lambda)$-approximation of $t(A,B,C)$ when 
$$\kappa_2 d k^2 \sqrt{t(G) } \log n \leq \lambda t(G) \Leftrightarrow t(G)\geq \frac{\kappa_2^2d^2k^4 \log ^2 n}{\lambda^2}.$$
In our algorithm, we apply Lemma~\ref{theo:sparse1} for $k=3$. Also, we will require $\theta=\frac{\eps}{6 \log n}$. So, the above sparsification lemma gives useful result in our algorithm when $t(A,B,C) \geq \frac{324 \kappa_2^2 d^2 \log^4 n}{\eps^2}$.
\end{itemize}

\end{rem}   
The algorithm sets a threshold $\tau=\max \left\{\frac{ 36\kappa_1^2 d^2 \log ^4 n}{\eps^2}, \frac{324 \kappa_2^2 d^2 \log^4 n}{\eps^2}\right \}$ and will proceed as follows:

\begin{description}

\item[Step 1:] {\bf (Threshold-Approx-Estimation)}
%\item[(1)]{\bf (Exact Counting)} 
%Fix the threshold $\tau $ as $\frac{ 36\kappa_1^2 d^2 \log ^4 n}{\eps^2}$, where $\kappa_1$ is the constant mentioned in Lemma~\ref{theo:sparse}. Then 
Run the algorithm {\sc Threshold-Approx-Estimation}, presented in Section~\ref{sec:coarse} with parameters $\tau$ and $\eps$. By Lemma~\ref{lem:exact_decide}, we either decide 
$t(G) > \tau$ or we have $\hat{t}$ which is an $(1 \pm \eps)$-approximation to $t(G)$. If $t(G) >\tau$, we go to {Step 2}. Otherwise, 
we terminate by reporting an estimate $\hat{t}$. The query complexity of {Step 1} is $\Oh\left(\frac{\tau \log^2 n}{\eps^2} \right)$.

\item[Step 2:] {\bf (General Sparsification)}
%\item[(2)] {\bf (General Sparsification)} 
$V(G)$ is \colored with $[3k]$ for $k=1$. Let $A, \, B, \, C$ be the partition generated 
by the coloring of 
$V(G)$. We initialize the data structure by setting $\psi =0$ and adding the tuple 
$(A, \, B, \, C, \, 9/2)$ to the data structure. Note that no query is required in this step. 
The constant $9/2$ is obtained by putting $k=1$ in Lemma~\ref{theo:sparse}.

\item[Step 3:]
We repeat {Steps 4} to {7} until there is no tuple left in the data structure. We maintain 
an invariant that the number of tuples stored in the data structure, is $\Oh(N)$, where $N=\frac{ \kappa_3\log^{12} n}{\eps^2}$. Note that $\kappa_3$ is a constant to be fixed later. %\complain{(Gopi: There is also an $N$ in Algorithm~\ref{algo:exact}. Please check.)}

\item[Step 4:] {\bf (Threshold for Tripartite Graph and Exact Counting in Tripartite Graphs)} 
For each tuple $(A,B,C,w)$ in the data structure, we determine whether $t(A,B,C) \leq
\tau $, the threshold, by using the deterministic algorithm 
corresponding to Lemma~\ref{lem:decide_exact} with 
$\cO(\tau \log n)$ many queries. If yes, we find the exact value of 
$t(A, \, B, \, C)$ by using the deterministic algorithm 
corresponding to Lemma~\ref{lem:exact} with 
$\cO(\tau \log n)$ many queries. Then the algorithm adds $w \cdot t(A, \, B, \, C)$ to $\psi$. 
We remove all $(A, \, B, \, C)$'s for which the algorithm found that $t(A, \, B, \, C)$ is below the threshold. 
As there are $ \Oh(N)$ many triples at any time, the 
number of queries made in each iteration of the algorithm is 
$\cO\left( \tau \log n \cdot N\right)=
\cO\left(\tau N \log n\right)$.

\item[Step 5:]
Note that each tuple $(A,B,C,w)$ in this step is such that 
$t(A,B,C) > \tau $. Let $(A_1,B_1,C_1,w_1),$ $\ldots,(A_r,B_r,C_r,w_r)$ 
be the set of tuples stored at the current instant. If $r >  10N$~\footnote{The constant $10$ is arbitrary. Any absolute constant more than $1$ would have been good enough.}, 
we go to {Step 6}. Otherwise, we go to {Step 7}. 

\item[Step 6] {\bf (Coarse Estimation and Sampling)} 
For each tuple $(A, \, B, \, C, \, w)$ in the data structure, we find an estimate $\widetilde{t}$ such that $\frac{t(A, \, B, \, C)}
{64 \log^2 n} < \hat{t} < 64t(A, \, B, \, C) \log^2 n$. This can be done due to Lemma~\ref{lem:coarse_main} and the number of queries is $\cO\left(\log^4 n \right)$ per tuple. As the algorithm executes the current step, the number of tuples in our data structure is more than $10N$. We take a sample from the set of tuples such that the sample maintains the required estimate \emph{approximately} by using Lemma~\ref{lem:importance}.
\remove{
\begin{lem}[\cite{BeameHRRS18}][Lemma~\ref{lem:importance} restated]
Let $(A_1, \, B_1, \, C_1, \, w_1),\ldots,(A_r, \, B_r, \, C_r, \, w_r)$ be the tuples present in the data structure and $e_i$ be the corresponding coarse estimation for $t(A_i, \, B_i, \, C_i), i \in [r],$  such that 
\begin{description}
\item[(i)] $\forall i \in [r]$, we have $w_i$, $e_i \geq 1$;

\item[(ii)] $\forall i \in [r]$, we have $\frac{e_i}{\rho} \leq t(A_i, \, B_i, \, C_i) \leq e_i \rho$ for some $\rho >0$; and

\item[(iii)] $\sum_{i=1}^r {w_i\cdot t(A_i, \, B_i, \, C_i)} \leq M$.
\end{description}
Note that the exact values $t(A_i, \, B_i, \, C_i)$'s are not known to us. Then there exists an algorithm that finds 
 $(A'_1, \, B'_1, \, C'_1, \, w'_1),
\ldots,(A'_s, \, B'_s, \, C'_s, \, w'_s)$ such that all of the above three conditions hold and  
$$
\size{\sum\limits_{i=1}^s {w'_i\cdot t(A'_i, \, B'_i, \, C'_i) } - \sum\limits_{i=1}^r {w_i\cdot t(A_i, \, B_i, \, C_i)}} \leq \lambda S,
$$ 
with probability $1-\delta$, where  $S=\sum_{i=1}^r {w_i\cdot t(A_i, \, B_i, \, C_i)}$ and $\lambda, \delta >0$. 
Also, 
$$
s=\cO\left(\lambda^{-2} \rho^4 \log M \left(\log \log M + \log \frac{1}{\delta}\right)\right).
$$
\end{lem} }
We use the algorithm corresponding to Lemma~\ref{lem:importance} with $\lambda =\frac{\eps}{6 \log n}$, $\rho=64 \log ^2 n$ and $\delta = \frac{1}{n^{10}}$ to find a new set of tuples $(A'_1, \, B'_1, \, C'_1, \, w'_1),
\ldots,(A'_s, \, B'_s, \, C'_s, \, w'_s)$ such that 
$$
\size{S-\sum_{i=1}^s w'_i ~ t(A', \, B', \, C')} \leq  \lambda S
$$
with probability $1-\frac{1}{n^{10}}$, 
where $S=\sum_{i=1}^r w_i t(A_i, \, B_i, \, C_i)$ and $s =   \frac{\kappa_3\log^{12} n}{\eps^2}$ for some constant $\kappa_3 >0$. This $\kappa_3$ is same as the one mentioned in Step $3$. Also, note that, $N=s=\frac{\kappa_3\log^{12} n}{\eps^2}$. No query is required 
to execute the algorithm of Lemma~\ref{lem:importance}. Recall that the number of tuples present at any time is $\Oh\left(N \right)$. Also, the coarse estimation for each tuple can be done by using $\Oh(\log ^4 n)$ many queries (Lemma~\ref{lem:coarse_main}). Hence, the number of queries in this step in each iteration, is $\cO(N \cdot \log ^{4} n)$. %\complain{(Gopi: Explain the $\log^4 n$ in a sentence.)}. 

\item[Step 7:] {\bf (Sparsification for Tripartite Graphs)}
We partition each of $A,B$ and $C$ into $3$ parts uniformly at random. Let f 
$A=U_1 \uplus U_2 \uplus U_3$; $V=V_1 \uplus V_2 \uplus  V_3$ and 
$W=W_1 \uplus W_2 \uplus W_3$. We delete $(A,B,C,w)$ from the data structure and 
add $(U_i,V_i,W_i,9w)$ for each $i \in [3]$ to our data structure. Note that no query is made in this step.

\item[Step 8:]
Report $\psi$ as the estimate for the number of triangles in $G$, when no tuples are left.

\end{description}

First, we prove that the above algorithm produces a $(1\pm \eps)$ multiplicative approximation to $t(G)$
for any $\eps > 0$ with high probability. Recall the description of {Step 1} of the algorithm. If the algorithm terminates in 
{Step 1}, then we have a $(1 \pm \eps)$ approximation to $t(G)$ by Lemma~\ref{lem:exact_decide}. Otherwise, we decide that $t(G)> \tau$ and 
proceed to {Step 2}. In {Step 2}, the algorithm colors $V(G)$ using three colors and incurs a multiplicative error of $1 \pm \eps_0$ to $t(G)$, where $\eps_0 = 
\frac{\kappa _1 d \log n}{\sqrt{t(G)}}$. This is because of Remark~\ref{rem:gen-spar} and our choice of $\tau$. As $t(G)> \tau$ and $n \geq 64$,  $ \eps_0 \leq 
\lambda = \frac{\eps}{6 \log n}$.
Note that the algorithm possibly performs {Step 4} to {Step 7} multiple times, 
but not more than $O(\log n)$ times, as explained below. 

Let $(A_1, \, B_1, \, C_1, \, w_1),\ldots, (A_{\zeta}, \, B_{\zeta}, \, C_{\zeta}, \, w_{\zeta})$ are the set of tuples present in the data 
structure currently. We define $\sum_{i=1}^{\zeta} t(A_i, \, B_i, \, C_i)$ as the number of \emph{active} triangles. Let 
$\act_i$ be the number of triangles that are active in the $i^{th}$ iteration. Note that $\act_1 \leq t(G) \leq n^3$. 
By Lemma~\ref{theo:sparse1} and {Step 7}, observe that $\act_{i+1} \leq \frac{\act_{i}}{2}$. So, after $3 \log n$ 
many iterations there will be at most constant number of active triangles and then we can compute the 
exact number of active triangles and add it to $\psi$. In each iteration, there can be a multiplicative 
error of $1 \pm \lambda$ in {Step 5} and $1 \pm \eps_0$ due to {Step 4}. So, using the fact that 
$\eps_0 \leq \lambda $, the multiplicative approximation factor lies between $(1-\lambda)^{3 \log n +1}$ 
and $(1+\lambda)^{3 \log n +1}$. As 
$\lambda= \frac{\eps}{6 \log n}$, the required approximation factor is $1 \pm \eps$. 

The query complexity of {Step 1} is $\cO\left(\frac{\tau \log n}{\eps^2} \right)$. {Steps} 2,  3,  5, 7 and 8 do not make any query to the oracle. The query complexity of {Step 4} is $\cO\left(\tau N \log n\right)$ in each iteration and that of {Step 6} is $\cO(N\log ^4 n)$ in each iteration. The total number 
of iterations is $\cO(\log n)$. Hence, the total query complexity of the algorithm is 
$$\cO\left({\eps^{-2}\tau \log n}+(\tau \log n + \tau N \log n + N\log ^4 n) \log n \right)=\cO(\eps^{-4}d^2 \log ^{18} n ).$$ 
In the above expression, we have put $\tau=\max \left \{\frac{ 36\kappa_1^2 d^2 \log ^4 n}{\eps^2}, \frac{324 \kappa_2^2 d^2 \log^4 n}{\eps^2}\right\}$  and $N=\frac{ \kappa_3\log^{12} n}{\eps^2}$.
 
 Now, we bound the failure probability of the algorithm. The algorithm can fail in
 {Step 1} with probability  at most $\frac{1}{n^{10}}$,
 {Step 2} with probability at most $\frac{2}{n^4}$,
 {Step 6} with probability at most 
$\frac{10\kappa_3 \log ^{12} n}{\eps^4}\cdot \frac{1}{n^9} + \frac{1}{n^{10}}$, and
 {Step 7} with probability at most $\frac{10\kappa_3 \log ^{12} n}{\eps^4}\cdot \frac{1}{n^8}$. As the algorithm 
might execute {Steps 4} to {6} for $3 \log n$ times, the total failure probability is bounded by
$$
\frac{1}{n^{10}}+\frac{2}{n^4}+3 \log n \left(\frac{10\kappa_3 \log ^{12} n}{\eps^4}\cdot \frac{1}{n^8} 
+ \frac{10\kappa_3 \log ^{12} n}{\eps^4}\cdot \frac{1}{n^9} +  \frac{1}{n^{10}}\right) \leq \frac{c}{n^2}.
$$
Note that the above inequality holds because $\eps > \frac{\sqrt{d}\log ^{9/2}n}{n^{3/4}}$ and $n \geq 64$.
  
 \remove{Finally, we state the result of counting $t(G)$, given $V(G)$ with $\Delta \leq d$, and an access to \tis oracle.
\begin{theo}
\label{theo:main}
Let $G$ be a graph with $\Delta \leq d$, $V(G)=[n]$ and $n \geq 64$\remove{ and we have an access to \tis oracle}. For any 
$\eps >0$, there exists an algorithm that computes $1\pm \eps$ approximation of $t(G)$ using 
$\cO_d(\frac{\log ^{24} n}{\eps^{12}})$ \tis queries.
\end{theo}}
We end this Section by restating our main result.
\begin{theo}[Restatement of Theorem~\ref{theo:main-restate}]
\label{theo:main-algo}
Let $G$ be a graph with $\Delta_E \leq d$, $\size{V(G)}=n \geq 64$. For any
$\eps >0$, \test can be solved using $\cO\left( \frac{d^{2}\log ^ {18} n}{\eps^{4}} \right)$ many \tis queries with probability {at least} $1-\frac{\cO(1)}{n^{2}}$.
\end{theo}

%% file: conclude.tex
\section{Discussions}
\label{sec:conclude}
\noindent
In this work, we generalize the framework of Beame et al~\cite{BeameHRRS18} of {\sc Edge Estimation} to solve \test by using \tis queries. Our algorithm makes $O(\eps^{-4}d^2\log ^{18} n)$ many \tis queries and returns a $(1\pm \eps)$-approximation to the number of triangles with high probability, where $d$ is the upper bound on $\Delta_E$.  
The downside of our work is the assumption $\Delta_E \leq d$. Note that Beame et al.~\cite{BeameHRRS18} had no such assumtion. Removing the assumption is non-trivial mainly due to the fact that, unlike the case for edges where two edges can share a common vertex, two triangles can share an edge. Our sparsification algorithm crucially uses the assumption on $\Delta_E$ and that remains the main barrier to cross. Recall our sparsification lemma (Lemma~\ref{theo:sparse}) and the definition of properly colored triangles (Definition~\ref{defi:proper}). Roughly speaking, our sparsification algorithm first colors the vertices of the graph, then counts the number of properly colored triangles, and finally scales it to have an estimation of the total number of triangles in the graph. Consider the situation when all the triangles in the graph have a common edge $e$. If $e$ is not properly colored, then we can not keep track of any triangle in $G$. As a follow up to this paper, Dell et al.\ ~\cite{DellLM20} and Bhattacharya et al.\ ~\cite{BhattaBGM18}, independently, generalized our result to $c$-uniform hypergraphs, where $c \in \N$ is a constant. In Section~\ref{sec:intro}, we already noted that \test can also be thought of as {\sc Hyperedge Estimation} problem in a $3$-uniform hypergrah. Their results showed that the bound on $\Delta_E$ is not necessary to solve \test by using polylogarithmic many \tis queries. The main technical result in both the papers is to come up with a sparsification algorithm that can take care of the case when $\Delta_E$ is not necessarily bounded. Note the sparsification algorithms in both the papers are completely different and give different insights.

Bhattacharya et al.~\cite{BhattaBGM18} and Dell et al.~\cite{DellLM20} refer the generalized oracle as {\sc Generalised Partite Independent Set} ({\sc GPIS}) oracle and {\sc Colorful Decision} ({\sc CD}) oracle, respectively. Bhattacharya et al.~\cite{BhattaBGM18} showed that {\sc Hyperedge Estimation} can be solved by using ${\cO}_c \left(\eps^{-4}{\log^{5c+5} n} \right)$ many  {\sc GPIS} queries and Dell et al.~\cite{DellLM20} showed that it can be solved  by using ${\cO}_c \left(\eps^{-2}{\log^{4c+8} n} \right)$ many {\sc CD} queries~\footnote{The constant in $\Oh_c(\cdot)$ is a function of $c$. The result of Bhattacharya et al. is a high probability result. The exact bound in the paper of Dell et al. is  ${\cO}_c \left(\eps^{-2}{\log^{4c+7} n} \log \frac{1}{\delta} \right)$, where the probability of success of their algorithm is $1-\delta.$}, with high probability. Substituting $c=3$ in their algorithm, we can have two different algorithms for \test. 
Let us compare our result (stated in Theorem~\ref{theo:main-algo}) with the results of \cite{BhattaBGM18} and Dell et al.~\cite{DellLM20} in the context of \test. If $\Delta_E=o(\log n)$, our algorithm for \test have less query complexity than that of Bhattacharya et al.~\cite{BhattaBGM18} for any given $\eps > 0$. Also, when $\Delta_{E}=o(\log n)$ and $\eps > 0$ is a fixed constant, our algorithm for \test have less query complexity than that of Dell et al.~\cite{DellLM20}. 

\remove{
\begin{theo}~\cite{BhattaBGM18}
\label{theo:gpis}
Let $\cH$ be a $c$-uniform hypergraph with $\size{U(\cH)}=n$. For any $\eps \in (0,1)$, {\sc Hyperedge Estimation} can be solved using 
${\cO}_c \left(\\eps^{-4}{\log^{5c+5} n} \right)$ {\sc GPIS} queries with high probability, where the constant in $\Oh_c(\cdot)$ is a function of $c$.
\end{theo}
\begin{theo}~\cite{BhattaBGM18}
\label{theo:gpis}
Let $\cH$ be a $c$-uniform hypergraph with $\size{U(\cH)}=n$. For any $\eps \in (0,1)$, {\sc Hyperedge Estimation} can be solved using ${\cO}_c \left(\\eps^{-2}{\log^{4c+7} n} \right)$ {\sc GPIS} queries with high probability, where the constant in $\Oh_c(\cdot)$ is a function of $c$.
\end{theo}
Putting $c=3$ in the above theorem, we can say that \test can be solved by using $O(\eps^{-1}\log ^{20} n)$ many \tis queries. Note that the query complexity is more than that of the query complexity of our algorithm stated in Theorem~\ref{theo:main}) when $=o(\log n)$. Dell et al.~\cite{} also have a result analogous to Theorem~\ref{theo:gpis} and the query complexity of their algorithm is ${\cO}_c \left(\\eps^{-1}{\log^{4c+7} n} \right)$. Putting $c=3$, we have \test can be solved by using $O(\eps^{-1}\log ^{19} n)$ many \tis queries. 
 However, we would like to note that our algoeNow we discuss about the main barrier in our calculation to remove the bound on $\Delta_E$ and how it is removed in our subsequent work in~\cite{BhattaBGM18}. 
\paragra
Note that the \test can also be thought of as {\sc Hyperedge Estimation} problem in a $3$-uniform hypergrah.}

%% file: appendix.tex
\section{Scenario where $\Delta_E$ is bounded}
\label{sec:deltabound}
In this Section, we discuss some scenarios where the number of triangles sharing an edge is bounded. An obvious example for such graphs are graphs with bounded degree. We explore some other scenarios.
\begin{itemize}
\item[(i)] Consider a graph $G(P,E)$ such that the vertex set $P$ corresponds to a subset of $\R^2$ and $
(u,v) \in E$ if and only if the distance between $u$ and $v$ is exactly $1$. The objective is to compute 
the number of triples of points from $P$ forming an equilateral triangle having side length $1$, that is, the number of triangles 
in $G$. Observe that there can be at most two triangles sharing an edge in $G$, that is, $\Delta_E \leq 2$.
\item[(ii)]  Consider a graph $G(P,E)$ such that the vertex set $P$ corresponds to a set of points inside an $N \times N$ square in $\R^2$ and $(u,v) \in E$ if and only if the distance between $u$ and $v$ is at most $1$. The objective is to compute 
the number of triples of points from $P$ forming a triangle having each side length at most $1$, that is, the number of triangles 
in $G$. For \emph{large} enough $N$ there can be bounded number of triangles sharing an edge in $G$ with high probability.
%\item Let us consider a tripartite graph $(A,B,C)$ such that $A, B, C \subset \N$ and $()$
\item[(iii)] Consider a graph $G(V,E)$ representing a community sharing information. Each node has some information and two nodes are connected if and only if there exists an edge between the nodes. Nodes increase their information by sharing information among their neighbors in $G$. Observe that the information of a node is derived by the set of neighbors. So, if two nodes have \emph{large} number of common neighbors in $G$, then there is no need of an edge between the two nodes. So, the number of triangles on any edge in the graph is bounded. The objective is to compute the number of triangles in $G$, that is, the number of triples of nodes in $G$ such that each pair of vertices are connected.  
\end{itemize}
In (i) and (ii), \tis oracle can be implemented very efficiently. We can report a \tis query by just running a standard plane sweep algorithm in Computational Geometry that takes $\Oh(n \log n)$ running time.

\input{prelim.tex}

%% file: prelim.tex
\section{Some probability results}
\label{sec:prelim}
\begin{pro}
\label{pro:exp}
Let $X$ be a random variable. Then $\E[X] \leq \sqrt{\E[X^2]}$.
\end{pro}
\begin{lem}{\rm (\cite[Theorem~7.1]{DubhashiP09})}{\bf .}
\label{lem:dp}
Let $f$ be a function of $n$ random variables $X_1,\ldots,X_n$ such that 
\begin{itemize}
\item[(i)] Each $X_i$ takes values from a set $A_i$,
\item[(ii)] $\E[f]$ is bounded, i.e., $0 \leq \E[f] \leq M$,
\item[(iii)] $\cB$ be any event satisfying the following for each $i \in [n]$.
$$ \size{\E[f ~|~X_1,\dots,X_{i-1},X_{i}=a_i,\cB^c] - \E[f ~|~X_1,\dots,X_{i-1},X_{i}=a'_i,\cB^c] }\leq c_i.$$
\end{itemize}  

Then for any $\delta \geq 0$, 
$$\pr\left(\size{f - \E[f]} > \delta + M\pr(\cB) \right) \leq e^{-{\delta^2}/{\sum\limits_{i=1}^{n} c_i^2}} + \pr(\cB).$$
\end{lem}

\begin{lem}[Hoeffding's inequality~\cite{DubhashiP09}]
\label{lem:hoeff_inq}
Let $X_1,\ldots,X_n$ be $n$ independent random variables such that $X_i \in [a_i,b_i]$. Then for $X=\sum\limits_{i=1}^n X_i$, the following is true for any $\delta >0$.
$$\pr \left( \size{X - \E[X]} \geq \delta \right) \leq 2 \cdot e ^{-{2\delta ^2}/{\sum\limits_{i=1}^n (b_i - a_i)^2}}.$$
\end{lem}

\begin{lem}[Chernoff-Hoeffding bound~\cite{DubhashiP09}]
\label{lem:cher_bound}
Let $X_1, \ldots, X_n$ be independent random variables such that $X_i \in [0,1]$. For $X=\sum\limits_{i=1}^n X_i$ and $\mu_l \leq \E[X] \leq \mu_h$, the followings hold for any $\delta >0$.
\begin{itemize}
\item[(i)] $\pr \left( X > \mu_h + \delta \right) \leq e^{-2\delta^2/n}$.
\item[(ii)] $\pr \left( X < \mu_l - \delta \right) \leq e^{-2\delta^2 / n}$.
\end{itemize}

\end{lem}

\begin{lem}\rm{(\cite[Theorem~3.2]{DubhashiP09})}{\bf .}
\label{lem:depend:high_exact_statement}
Let $X_1,\ldots,X_n$ be random variables such that $a_i \leq X_i \leq b_i$ and $X=\sum\limits_{i=1}^n X_i$. Let $\cD$ be the \emph{dependent} graph, where $V(\cD)=\{X_1,\ldots,X_n\}$ and $ E(\cD)= \{(X_i,X_j): \mbox{$X_i$ and $X_j$ are dependent}\}$. Then for any $\delta >0$, 
$$ \pr(\size{X-\E[X]} \geq \delta) \leq  2e^{-2\delta^2 / \chi^*(\cD)\sum\limits_{i=1}^{n}(b_i-a_i)^2},$$
where $\chi^*(\cD)$ denotes the \emph{fractional chromatic number} of $\cD$.

\end{lem}

The following lemma directly follows from Lemma~\ref{lem:depend:high_exact_statement}.
\begin{lem}
\label{lem:depend:high_prob}
Let $X_1,\ldots,X_n$ be indicator random variables such that there are at most $d$ many $X_j$'s on which an $X_i$ depends and  $X=\sum\limits_{i=1}^n X_i$. Then for any $\delta > 0$, $$\pr(\size{X-\E[X]} \geq \delta) \leq 2e^{-2\delta^2 / (d+1)n}.$$
\end{lem}

\begin{lem}[Importance sampling~\cite{BeameHRRS18}]
\label{lem:importance1}
Let $(D_1,w_1,e_1),\ldots, (D_r,w_r,e_r)$ are the given structures and each $D_i$ has an associated weight
${c}(D_i)$ satisfying  
\begin{itemize}
\item[(i)] $w_i,e_i \geq 1, \forall i \in [r]$;
\item[(ii)] $\frac{e_i}{\rho} \leq c(D_i) \leq e_i \rho$ for some $\rho >0$ and all $i \in [r]$; and
\item[(iii)] $\sum\limits_{i=1}^r {w_i\cdot c(D_i)} \leq M$.
\end{itemize}
Note that the exact values $c(D_i)$'s are not known to us. Then there exists an algorithm that finds 
 $(D'_1,w'_1,e'_1),\ldots, (D'_s,w'_s,e'_s)$ such that, with
 probability at least $1-\delta$, all of the above three conditions hold and 
$$
 \size{\sum\limits_{i=1}^t {w'_i\cdot c(D'_i)} - \sum\limits_{i=1}^r
   {w_i\cdot c(D_i)}} \leq \lambda S,
$$ 
where  $S=\sum\limits_{i=1}^r {w_i\cdot c(D_i)}$ and $\lambda, \delta
>0$. 
The time complexity of the algorithm is $\cO(r)$ 
and $s=\cO\left( \frac{\rho^4 \log M \left(\log \log M + \log \frac{1}{\delta}\right)}{\lambda^2}\right)$.
\end{lem}